\documentclass[acmsmall]{acmart}

\AtBeginDocument{%
  \providecommand\BibTeX{{%
    \normalfont B\kern-0.5em{\scshape i\kern-0.25em b}\kern-0.8em\TeX}}}

\usepackage{amsmath,amssymb,amsfonts, amsthm}
\usepackage{algorithm}
\usepackage{algcompatible}
\usepackage[noend]{algpseudocode}
\usepackage{graphicx}
\usepackage{textcomp}
\usepackage{xcolor}
\usepackage{threeparttable}
\usepackage{tabularx}
\usepackage{multirow}
\usepackage{color}
\definecolor{Gray}{gray}{0.9}
\usepackage{subfloat}
\usepackage{setspace}
\usepackage{hyperref}
\usepackage{bbm}
\usepackage{adjustbox}
\newcolumntype{g}{>{\columncolor{Gray}}c}
\newcolumntype{w}{>{\columncolor{white}}c}

\DeclareMathOperator{\sgn}{sgn}

\newcommand{\qq}[1]{\textcolor{black}{#1}}
\usepackage{colortbl}

\newtheorem{lemma}{Lemma}
\newcounter{noteQWctr} 

\usepackage[caption=false,font=footnotesize,labelfont=sf,textfont=sf]{subfig}

\makeatletter
\newcommand{\multiline}[1]{%
  \begin{tabularx}{\dimexpr\linewidth-\ALG@thistlm}[t]{@{}X@{}}
    #1
  \end{tabularx}
}
\makeatother

\begin{document}

\title{BATS: A Spectral Biclustering Approach to Single Document Topic Modeling and Segmentation}

\author{Qiong Wu}
\email{qwu05@email.wm.edu}
\affiliation{%
  \institution{Department of Computer Science, College of William and Mary}
  \streetaddress{200 Stadium Dr}
  \city{Williamsburg}
  \state{Virginia}
  \postcode{23185}
}

\author{Adam Hare}
\email{adam.hare@zoomi.ai}
\affiliation{%
  \institution{Zoomi Inc.}
  \streetaddress{325 Sentry Parkway, Suite 200}
  \city{Blue Bell}
  \state{Pennsylvania}
  \postcode{19422}
}

\author{Sirui Wang}
\email{swang23@email.wm.edu}
\affiliation{%
  \institution{Department of Computer Science, College of William and Mary}
  \streetaddress{200 Stadium Dr}
  \city{Williamsburg}
  \state{Virginia}
  \postcode{23185}
}

\author{Yuwei Tu}
\email{yuwei.tu@zoomi.ai}
\affiliation{%
  \institution{Zoomi Inc.}
  \streetaddress{325 Sentry Parkway, Suite 200}
  \city{Blue Bell}
  \state{Pennsylvania}
  \postcode{19422}
}

\author{Zhenming Liu}
\email{zliu20@wm.edu}
\affiliation{%
  \institution{Department of Computer Science, College of William and Mary}
  \streetaddress{200 Stadium Dr}
  \city{Williamsburg}
  \state{Virginia}
  \postcode{23185}
}

\author{Christopher G. Brinton}
\email{cgb@purdue.edu}
\affiliation{%
  \institution{School of Electrical and Computer Engineering, Purdue University}
  \streetaddress{610 Purdue Mall}
  \city{West Lafayette}
  \state{Indiana}
  \postcode{47907}
 }

\author{Yanhua Li}
\email{yli15@wpi.edu}
\affiliation{%
  \institution{Department of Computer Science, Worchester Polytechnic Institute}
  \streetaddress{100 Institute Rd}
  \city{Worcester}
  \state{Massachusetts}
  \postcode{01609}
 }


\begin{abstract}
Existing topic modeling and text segmentation methodologies generally require large datasets for training, limiting their capabilities when only small collections of text are available. In this work, we reexamine the inter-related problems of ``topic identification" and ``text segmentation" for sparse document learning, when there is a single new text of interest. In developing a methodology to handle single documents, we face two major challenges. First is \textit{sparse information}: with access to only one document, we cannot train traditional topic models or deep learning algorithms. Second is \textit{significant noise}: a considerable portion of words in any single document will produce only noise and not help discern topics or segments. To tackle these issues, we design an unsupervised, computationally efficient methodology called BATS: Biclustering Approach to Topic modeling and Segmentation. BATS leverages three key ideas to simultaneously identify topics and segment text: (i) a new mechanism that uses word order information to reduce sample complexity, (ii) a statistically sound graph-based biclustering technique that identifies latent structures of words and sentences, and (iii) a collection of effective heuristics that remove noise words and award important words to further improve performance. Experiments on \qq{six} datasets show that our approach outperforms several state-of-the-art baselines when considering topic coherence, topic diversity, segmentation, and runtime comparison metrics.
\end{abstract}

\vspace{-2mm}

\maketitle

\section{Introduction}\label{sec:introduction}
\vspace{-.6mm}
Innovations in topic modeling and text segmentation have demonstrated the potential for automated analyses of large collections of documents. Broadly speaking, \textit{topic modeling} refers to finding a collection of topics (e.g., groups of words) that represent a given document, whereas \textit{document segmentation} refers to partitioning a document into components (e.g., sentences) about the topics. Existing solutions to these problems are usually based on analyzing statistical patterns in text across datasets that consist of large collections of documents. For example, the popular Latent Dirichlet Allocation (LDA) algorithm for topic modeling \cite{blei2003latent} assumes that each document comprising a corpus, and every word in them, are generated according to the Dirichlet process. With this assumption, EM-based algorithms can then be employed to infer the latent states of the documents~\cite{hofmann2017probabilistic}. Word embedding models such as word2vec~\cite{mikolov2013distributed} and GloVe \cite{pennington2014glove} have also become popular, building joint distributions of word sequences by transforming every word in a document into a high-dimensional space learnt over a large corpus. The resulting high-dimensional representations then help to identify topics in the document and perform segmentation based on these topics.


While algorithms for finding topics \cite{blei2003latent, hofmann2017probabilistic, deerwester1990indexing} and segmenting documents \cite{hearst1997texttiling, riedl2012topictiling, choi2000advances} have been extensively studied, none have fully addressed issues posed by the ``new and single document'' setting. In this setting, we may need to analyze a newly created text whose topics have not been seen before, posing unique modeling challenges we aim to address in this paper. 

\vspace{-2mm}
\subsection{Motivation: The New and Single Document Setting}
\label{ssec:motivation}
\vspace{-.6mm}
The ``new and single document'' setting manifests itself in several contemporary scenarios. We consider first that \textit{new words may arise rapidly and garner the most interest when there is relatively little written about them} \cite{hong2010empirical,zhao2011comparing}. A topic modeling algorithm that relies on a model trained on a dataset from several years ago may have rejected a topic word such as ``COVID'' as out-of-vocabulary at the moment when the public was most interested in finding out about the emerging disease. Although a news aggregator like Google may be able to find a suitable number of documents despite so few being available, a model used by a single outlet (e.g., an internal search on the \textit{New York Times} website) is unlikely to have the same resources.

These neologisms are not the only circumstance in which the ``new and single document'' setting arises. Often, \textit{existing words acquire new context-specific meanings as their usage changes over time} \cite{yao2018dynamic}. For example, before the ubiquity of ``like buttons'' on social media, it may have been safe to assume the word ``like” was too generic to be a suitable topic or topic word. Additionally, \textit{words may have different meaning in different domains}:  consider ``transformer” in the context of computer science \cite{vaswani2017attention}, electrical engineering, and popular culture. Although this issue has helped spark interest in contextualized word embeddings \cite{peters2017semi,mccann2017learned,devlin2018bert}, these models are still sensitive to biases in the historical data used to train them \cite{zhao2019gender}. These differences may be especially relevant as one source begins to add content from a new domain, e.g., when an online archive begins accepting submissions from a new subject area. In each of these settings, when working across a diverse database of documents, treating each one individually -- i.e., as a new and single document -- can allow these semantic differences to be better captured. The same could be said of a general-purpose aggregator that can receive documents in new formats and, as a result, cannot afford to make strong assumptions about future input based on what it has already seen. A key challenge we face as a result of the single document setting is dealing with sparser amounts of text available for modeling.

The ``new and single document'' setting also manifests itself \textit{when rapid processing is required in the presence of constrained computational resources}. These scenarios are becoming more widespread today as a result of edge computing technologies \cite{shi2016edge,tran2017collaborative}, which are moving processing tasks from the cloud to the network edge in an effort to leverage the improved intelligence capabilities of mobile devices for more real-time results. Edge devices are heterogeneous in their processing capabilities, which requires computationally efficient learning techniques \cite{chiang2016fog,tu2020network}. In the text analysis space, we can consider continually-updated document streams (e.g., a social media newsfeed) that need to be processed in real time with limited resources (e.g., on a smartphone), which will require efficient algorithms that avoid computation across a large corpus \cite{yao2009efficient}. When speed is a concern and parallelization is a possibility, this approach allows each document to be treated independently and pushed to the next step in the pipeline as it is processed. It also avoids relying on a large corpus or pre-trained model, both of which may demand substantial computational resources.

\vspace{-2mm}
\subsection{BATS: Objectives and Key Techniques}
\vspace{-.6mm}
In this paper, we design a statistically sound, computationally efficient, unsupervised algorithm that can simultaneously extract topics and segment text from a single document of interest. Designing such an algorithm is challenging because we need to determine model parameters on a sparse dataset. Our development is guided by three key ideas:

\vspace{-2mm}
\subsubsection{Idea 1: Using word ordering information properly}
Traditional topic modeling approaches assume bag-of-words models \cite{blei2003latent} where information on the order in which words appear is neglected. While this has proven effective in the analysis of full corpora, compression to a bag-of-words in the case of a single document may lose information valuable to the task at hand. The recent success of recurrent models and the addition of positional encodings in non-recurrent models for the application of machine translation \cite{vaswani2017attention} is further evidence of the potential value of word-order information on single document.

Motivated by this, our approach aims to leverage word-order information to achieve good performance in the presence of a small, single document training dataset. In particular, we consider the location of words in neighboring sentences. In designing this mechanism, we will make two assumptions guided by basic rules of written language: (i) words appearing in the same sentences are more likely to be on the same topic, and (ii) words located in nearby sentences are more likely to be on the same topic.


\vspace{-2mm}
\subsubsection{Idea 2: Design a biclustering algorithm that addresses sparsity}
For joint topic modeling and text segmentation, we will find it convenient to model documents with sentence-word matrices. But word-to-word interactions and word-to-sentence interactions are noisy by nature \cite{brinton2018efficiency}. This problem becomes even more pronounced with small datasets like single documents where these interactions are likely to be sparse (e.g., the sentence-word matrices for datasets considered in this paper have only 15\% of entries nonzero on average). A well-designed denoising process is necessary so that a sentence-word matrix can be utilized effectively in the downstream topic extraction and text segmentation tasks. \qq{We also seek to avoid reliance on pre-trained word embeddings in this step given the computational cost consideration in our ``single and new document'' setting.}

Our approach connects the denoising problem here with the denoising problem in stochastic block models, \qq{where we consider structure in the ``blocks'' of a document's sentence-word matrix}~\cite{qin2013regularized,wu2019adaptive}. In particular, \qq{motivated by an expected power law of node degrees in a document's sentence-word bipartite graph,} we design a specialized spectral biclustering algorithm which operates on a regularized version of the graph Laplacian to address sparsity. In this algorithm, the topics and segments emerge from clustering the right and left singular vectors of the Laplacian. Given this, we term our overall solution \textbf{BATS}: \textbf{B}iclustering \textbf{A}pproach for \textbf{T}opic modeling and \textbf{S}egmentation.

\vspace{-2mm}
\subsubsection{Idea 3: Optimize heuristics to analyze single documents}
We design several heuristics to enhance our algorithm's performance. Our heuristics are designed based on two major observations: (i) extremely low frequency words tend to introduce noise to document analysis and thus need to be removed, and (ii) part-of-speech (POS) tagging can help to identify more important elements of a document and thus should be considered in our model. Therefore, we remove the low frequency words in the text, but award the important words according to their POS tags. Specifically, because nouns and verbs often convey the body and condition of a sentence, they are typically more informative in topic modeling than other parts of speech \cite{goldberg2012dynamic}. 

\vspace{-2mm}
\subsubsection{Experimental validation}
We evaluate BATS against 12 total baselines, six for topic modeling and six for text segmentation tasks. For topic modeling, we compare performance in terms of topic coherence (i.e., quality of individual topics) and topic diversity (i.e., overlap in topic words) on \qq{five} datasets, in which we find that BATS obtains comparable or higher performance to the best baselines in each case. For text segmentation, we add in one more standard dataset, and show that we outperform all of the baselines in most cases in terms of agreement with a ground truth. \qq{In most cases, the highest performing baselines on topic coherence and text segmentation are based on pre-trained lanaguage models. To this end, we additionally show that the runtime of our method scales significantly better than any of the highest-performing baselines as the size of the input document increases.} For these reasons, we can conclude that BATS obtains the most desirable combination of metrics for the ``new and single document'' setting.

\vspace{-2mm}
\subsection{BATS: Architecture and Roadmap}
\vspace{-.6mm}
Figure \ref{fig:framework} outlines the methodology we develop and provides a roadmap for the paper. The inputs to BATS are a single document and a single hyperparameter (segment number, which also indicates topic number). Then, the two major stages of BATS are preprocessing and extraction. In the preprocessing stage (Sections \ref{ssec:preprocessing} and \ref{ssec:laplacian}), we leverage ideas 1 and 3 to build an effective feature matrix representation of a document under sparse and noisy conditions. In the extraction stage (Sections \ref{ssec:singular_vectors} and \ref{ssec:bicluster}), we use idea 2 to identify low-dimensional representations of the signals through spectral biclustering, with agglomerative methods to segment the text and KMeans to identify the topics. Our subsequent evaluation (Section \ref{sec:experiment}) assesses performance of the resulting text segments and topic words in terms of diversity, coherence, segmentation, and runtime metrics.

\begin{figure*}[t]
    \vspace{-2mm}
    \centering
    \includegraphics[width=\linewidth]{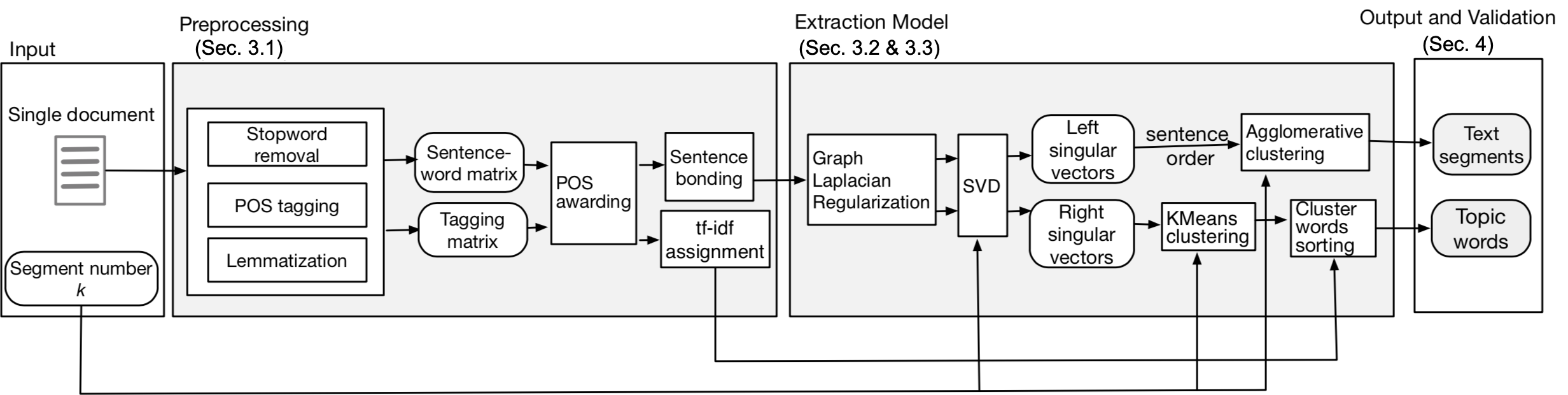}
    \caption{\small{Block diagram summary of the modules comprising BATS, the spectral biclustering methodology we develop in this paper for single-document topic modeling and text segmentation.}}
    \label{fig:framework}
    \vspace{-5.5mm}
\end{figure*}

\vspace{-2mm}
\subsection{Summary of Contributions}
\vspace{-.6mm}
Our key contributions are summarized as follows:
\vspace{-.6mm}
\begin{itemize}
\item We develop a novel methodology called BATS that performs topic modeling and text segmentation on a single document simultaneously. BATS is unsupervised and scalable in its implementation, as it does not rely on pre-trained word embedding models.

\item We connect the joint topic extraction and segmentation problem to spectral biclustering of sentence-word matrices, and show how a factorization of the graph Laplacian with appropriate pre-processing and post-clustering can lead to effective topic modeling and text segmentation.


\item Our evaluation on several datasets establishes that BATS achieves the best combination of topic modeling, text segmentation, and runtime metrics when compared with baselines on single documents. It also verifies the contribution of different components of BATS.



\end{itemize}
\section{Related Work}
\label{sec:related}
\vspace{-.6mm}
We identify three areas of related work: biclustering techniques, topic modeling, and text segmentation algorithms.


\vspace{-2mm}
\subsection{Biclustering Techniques}
\vspace{-.6mm}
Biclustering techniques (e.g., \cite{dhillon2001co,dhillon2003information,rohe2012co}) have been proposed to model interactions among two types of nodes represented in a bipartite graph, with nodes of each type grouped into clusters according to different methods. These techniques are widely used in part because of their sound theoretical properties \cite{dhillon2003information}.
In \cite{dhillon2001co} and \cite{kluger2003spectral}, the authors propose algorithms which translate input data into bipartite graphs and apply spectral techniques to the adjacency matrices; in \cite{dhillon2001co}, a block diagonal structure is assumed, while in \cite{kluger2003spectral}, the case of a checkerboard pattern is considered, with implications to the spectral decomposition.
\cite{rohe2012co} can be viewed as an extension of the algorithm in \cite{dhillon2001co} to deal with asymmetric data matrices. By contrast, \cite{dhillon2003information} proposes a probabilistic approach to graph biclustering, where the input data matrix is treated as a joint probability distribution between two random variables, which are then clustered according to relative entropy and mutual information metrics. Our work builds off the spectral clustering foundations in \cite{dhillon2001co,rohe2012co}, accommodating rectangular sentence-word data matrices instead of traditionally assumed square matrices.

\vspace{-2mm}
\subsection{Topic Modeling}
\vspace{-.6mm}
Several models have been proposed to extract topics from a corpus consisting of multiple long documents, including Latent Semantic Analysis (LSA) \cite{deerwester1990indexing},  Non-negative Matrix Factorization (NMF) \cite{pauca2004text}, Probabilistic Latent Semantic Analysis (pLSA) \cite{hofmann2017probabilistic}, Latent Dirichlet Allocation (LDA) \cite{blei2003latent}, and variants on LDA, e.g., hierachical modeling \cite{teh2005sharing} (see~\cite{daud2010knowledge} for a survey). Recent work has incorporated word and document embeddings jointly to capture ``topic vectors" \cite{angelov2020top2vec}; however, even when these approaches use large pre-trained embedding models, they require corpora on the order of thousands of documents to achieve competitive results, making them unsuitable to applications where data is limited.
Analysis on short texts usually faces the issue of sparsity in word occurrences. To overcome this challenge, works such as \cite{yin2014dirichlet,yan2013biterm} make additional assumptions on word co-occurrence patterns; 
\cite{shi2018short,nguyen2015improving, zhao2017metalda, qiang2017topic} have resorted to word embeddings which leverage pre-trained models; \cite{cui2019short, gao2017finding} depend on further external knowledge including social relationships in microblogs and user preferences. One of these, Semantics-assisted NMF (SeaNMF) \cite{shi2018short}, is a variant of NMF designed to handle short documents by learning semantic relationships between words in the corpus. It has been found to outperform other standard techniques such as NMF and LDA \cite{shi2018short}.

Different from these methods, ours aims at identifying topics in a single, newly created document without an extensive training component. To overcome issues of input data sparsity and noise, BATS turns to word-ordering information between sentences and regularization in the spectral clustering phase, as opposed to making additional assumptions on word co-occurrence patterns. Through evaluation on several datasets, we show that BATS outperforms many of the above models on single document topic modeling in terms of topic coherence, topic diversity, and scalability metrics. In particular, compared to the state-of-the-art SeaNMF method, we will see through our experiments that BATS achieves better performance on topic coherence in most cases, and smaller/more scalable runtimes, which is important in the ``new and single document'' setting we consider in this paper.

\vspace{-2mm}
\subsection{Text Segmentation}
\vspace{-.6mm}
Text segmentation algorithms are designed to detect breakpoints in a document and split the document into multiple segments accordingly. 
Algorithms such as Lexical Chains \cite{manabu1994word} and TextTiling \cite{hearst1997texttiling} use lexical co-occurrence and distribution patterns to divide sets of paragraphs into multi-paragraph sub-blocks that become segments. A potential drawback of these approaches, however, is that the segments are not associated or labeled with explicit topic information, and that it is not always clear how to translate from a lexical distribution to topics. This motivates the consideration of topic modeling and text segmentation jointly.

More recently, to improve segmentation performance, topic-based segmentation methods such as TopSeg \cite{brants2002topic}, LDA\_MDP \cite{misra2011text}, and TopicTiling \cite{riedl2012topictiling} have been proposed. Similar to the topic modeling algorithms discussed above, these segmentation methods depend heavily on the training process, and usually require training on a large corpus \cite{chang2009reading}. This is problematic when only small datasets are available, let alone in the single document case that we consider in this paper. Through biclustering of the sentence-word matrix and development of other pre-processing techniques, BATS does not demand an expensive training process. Other recent approaches, such as SupervisedSeg \cite{koshorek2018text} and SegBot \cite{li2018segbot}, have modeled the task as a supervised learning problem and employed deep recurrent networks while others rely on massive language representation models such as Bidirectional Encoder Representations from Transformers (BERT) \cite{devlin2018bert}. Generally, these models are pre-trained on large historical datasets and can be used in an application with little to no fine-tuning required. However, loading such models may use substantial computational resources and their reliance on historical data can introduce biases and gaps. By contrast, BATS is designed to be computationally efficient and flexible enough to accurately handle novel text. We find that BATS obtains an order of magnitude smaller runtime compared with BERT applied to text segmentation, which is important in our ``new and single document'' setting. Further, our evaluation shows that BATS outperforms the segmentation methods discussed here on single documents across several datasets. 

\vspace{-1mm}
\section{Spectral Biclustering Methodology}
\label{sec:methodology}
\vspace{-.6mm}
As shown in Figure \ref{fig:framework}, our proposed methodology BATS consists of two main stages: the text preprocessing stage (Section \ref{ssec:preprocessing}) and the extraction stage, with the latter broken down into graph Laplacian regularization (Section \ref{ssec:laplacian}), \qq{singular vector extraction (Section \ref{ssec:singular_vectors}),} and sentence/word clustering (Section \ref{ssec:bicluster}). Topics and segments emerge from the word and sentence clusters, respectively. In this section, we detail the development of these modules, and discuss how they address the issue of sparsity associated with modeling a single document. 

\vspace{-2mm}
\subsection{Document Preprocessing and Matrix Construction}
\label{ssec:preprocessing}
\vspace{-.6mm}
Consider an input document comprised of $m$ sentences, indexed $i = 1,...,m$. We denote $\mathcal{W} = \{w_1, ..., w_n\}$ as the set of words we are interested in for modeling, indexed $j = 1,...,n$. In defining $\mathcal{W}$, we do not include all the words that ever appear in the document; instead, a word is included in $\mathcal{W}$ if and only if it appears in more than one sentence in the document and it is not in a stopword list. In this way, $\mathcal{W}$ excludes ``degree-one'' words that can skew models in single documents; we observe that these words often behave as pure noise in our inference algorithms.



Let $X = [X_{ij}] \in \mathbb{R}^{m \times n}$ denote the sentence-word matrix. Even after excluding degree-one words, we still expect this matrix will be sparse, with a significant number of $X_{ij} = 0$. We develop two steps to construct $X$, taking into account both word order and parts-of-speech information:

\vspace{-2mm}
\subsubsection{Step 1. Using parts-of-speech information}
Our first optimization trick is based on parts-of-speech (POS) tags, which are generated through analysis of the word positions in the sentences \cite{goldberg2012dynamic}. In particular, the lexical model presented in \cite{elman1990finding} shows hierarchies exist according to syntactic/semantic similarities of words; looking into them, it is clear that nouns and verbs convey more information than other word types, and thus should be given a larger weight \cite{resnik1993selection}. As a result, letting $X^o = [X^o_{ij}]$ where $X^o_{ij}$ is the number of occurrences of word $w_j \in \mathcal{W}$ in sentence $i$, we define
\begin{equation} \label{eq:pos_awarding}
X^a = X^o + \lambda T ,
\end{equation}
where $T = [T_{ij}]$, $T_{ij} = 1$ if $X^o_{ij} \neq 0$ and $w_j$ is tagged as a noun or verb in sentence $i$, and $T_{ij} = 0$ otherwise. $\lambda > 0$ is a scalar parameter for awarding POS; by default, $\lambda = 1$. In our implementation, Python's {\tt spaCy} module is used to tag the words, as this pre-trained model based on word positions is more robust to novel words or topics than would be, for instance, a word-embedding model.

\vspace{-2mm}
\subsubsection{Step 2. Transformation by using word-order information}
Our incorporation of word-order information is based on the intuition that words in neighboring sentences are likely to be similar in their constituent topics, with this effect decaying as the sentences grow further apart. Assumptions on words appearing within a certain window being related can be found in other text analysis techniques as well, including word embedding models \cite{pennington2014glove}. Concretely, we bond neighboring sentences to the current sentence according to
\begin{equation} \label{eq:sentence_bonding}
X_i = \sum^{i+w}_{\ell=i-w} d^{|\ell-i|} X_\ell^a, \qquad i = 1, ..., m,
\end{equation}
where $X_i = (X_{i1}, ..., X_{in})$ is the $i$th row of $X$ and $X_\ell^a$ is the $\ell$th row of $X^a$ for $\ell = 1,...,m$ (for $\ell < 1$ and $\ell > m$, $X^a_\ell$ is taken as a vector of zeros). Parameter $w$ controls the size of the bonding window, and $d \in [0, 1]$ is a decay rate for the distance. In this way, the presence of a word in one sentence will impact neighboring sentences, and words appearing in several consecutive sentences are increased in importance. Doing so also alleviates the issue of sparsity associated with modeling single documents, as each sentence's data smoothens its neighbors' representations too. The procedure for determining the values of $w$ and $d$ will be discussed in Section \ref{ssec:laplacian}.

\vspace{-2mm}
\subsection{Graph Laplacian and Regularization}
\label{ssec:laplacian}
\vspace{-.6mm}

\subsubsection{\qq{Intuition}}
Consider the bipartite graph $\mathcal{G}(X)$ of the sentence-word matrix $X$, where the sentences $i = 1,...,m$ and words $j = 1,...,n$ each form a node set, and edge $(i,j)$ of weight $X_{i,j}$ is in $\mathcal{G}(X)$ if and only if $X_{i, j} \neq 0$. \qq{Intuitively, this graphical representation will capture some relatedness information between words and sentences, e.g., words that frequently occur together, and are therefore likely to be part of the same topic, should have a high number of short paths to each other. This bipartite graph is likely to have nodes whose degrees follow a power law distribution, both as a product of Zipf's Law for word frequency distributions~\cite{zipf1949human, zipf1953psycho} and because empirically real world networks often exhibit power laws~\cite{eikmeier2017revisiting}. The graph Laplacian is known to be a useful matrix representation for clustering graphs with heterogeneous node degrees given the relationship between its spectral properties and the connected components of the graph~\cite{von2007tutorial}. To further address the sparsity issue, we will develop a regularized version of the Laplacian, as it has been shown to improve the performance of spectral clustering algorithms on sparse matrices~\cite{amini2013pseudo, chaudhuri2012spectral, qin2013regularized}.} 

\vspace{-2mm}
\subsubsection{Constructing the Laplacian}
Formally, define two diagonal matrices $P = \mbox{diag}(P_1, ..., P_n) \in \mathbb{R}^{n \times n}$ and $O = \mbox{diag}(O_1,...,O_m) \in \mathbb{R}^{m \times m}$ where $P_j = \sum_{i=1}^m X_{ij}, j = 1,...,n$ and $O_i = \sum_{j=1}^n X_{ij}, i = 1,...,m$ are the row and column sums of $X$. With regularization parameters $\tau_p, \tau_o \geq 0$, the regularized graph Laplacian $L \in \mathbb{R}^{m \times n}$ is computed according to
\begin{equation} \label{eq:laplacian}
        P_{\tau} = P + \tau_p I_p, \qquad O_{\tau} = O + \tau_o I_o, \qquad L = (O_{\tau})^{-\frac{1}{2}} X (P_{\tau})^{-\frac{1}{2}},
\end{equation}
where $I_p$ and $I_o$ are identity matrices. Multiplying these by regularization parameters $\tau_p$ and $\tau_o$ can resolve issues due to poor concentration since the degrees for every vertex are inflated. Following prior work \cite{qin2013regularized} which has indicated that such regularization parameters should be proportional to the average degrees of the vertices (so that the asymptotic bounds will be indicative of the mis-clustering rate), we set the average degrees as defaults, i.e., $\tau_p = \sum_j P_{j} / n$ and $\tau_o = \sum_i O_i / m$.

\setlength{\textfloatsep}{0.1cm}
\begin{algorithm}[t]
\small
\setstretch{0.8}
\caption{Matrix decomposition on regularized Laplacian.}
\bgroup
\def\arraystretch{1.2}
\begin{algorithmic}[1]
    \Statex \textbf{INPUT:} Original sentence-word matrix $X^o$, POS-based matrix $T$
    \Statex \textbf{PARAMETER:} Awarding value $\lambda$, window size $w$, decaying rate $d$, segment number $k$
    \Statex \textbf{OUTPUT:} Matrix $U$ for sentences and $V^T$ for words
    \Function{MAT\_DECOMP}{$X^o$, $w$, $d$}
        \If{$\lambda > 0$}
            \State $X^a = X^o + \lambda T$ \qquad  //Word awarding
        \Else
            \State $X^a = X^o$
        \EndIf
        \State $F \gets $ tf-idf($X^a$) \qquad //Tf-idf assignment
        \For{$i \gets 1,...,n$}
            \State $X_i = \sum^{i+w}_{\ell=i-w} d^{|\ell-i|} X_\ell^a$ \qquad //Sentence bonding
        \EndFor
        \For{$j \gets 1,...,n$}
            \State $P_j \gets \sum_{i=1}^m X_{ij}$
        \EndFor
        \For{$i \gets 1,...,m$}
            \State $O_i \gets \sum_{j=1}^n X_{ij}$
        \EndFor
        \State $\tau_p \gets \sum_j P_{j} / n$, $P_{\tau} \gets P + \tau_p I_p$ \qquad //Regularization
        \State $\tau_o \gets \sum_i O_i / m$, $O_{\tau} \gets O + \tau_o I_o$ \quad //Regularization
        \State $L = (O_{\tau})^{-\frac{1}{2}} X (P_{\tau})^{-\frac{1}{2}}$ \qquad //Graph Laplacian
        \State $U \Sigma V^T = L$ \qquad  //Singular value decomposition
        \For{$u' \gets$ rows of $U$}
        \State $u' \gets u'[1:k]$ \qquad //Reserve first $k$ dimensions
        \State $u' \gets u' / \sqrt{\sum_{i}u_i^{'2}}$ \qquad //L2 normalization on $U'$
        \EndFor
        \For{$v' \gets$ rows of $V$}
        \State $v' \gets v'[1:k]$ \qquad //Reserve first $k$ dimensions
        \State $v' \gets v' / \sqrt{\sum_{i}v_i^{'2}}$ \qquad //L2 normalization on $V'$
        \EndFor
        \State \textbf{return} $U'$, $V'$, $F$ \quad  //$U'$ for sentences, $V'$ for words
    \EndFunction
\end{algorithmic}
\egroup
\label{algo1}
\end{algorithm}

\vspace{-2mm}
\subsection{\qq{Sentence and Word Singluar Vectors}}
\label{ssec:singular_vectors}
\vspace{-.6mm}

\subsubsection{\qq{Intuition}}
\qq{Our main observation is that sentence-word interactions in a document may be modeled by a stochastic co-block model (SCBM)~\cite{karrer2011stochastic,qin2013regularized}, whose structure can be inferred through spectral clustering. SCBM is a bipartite graph generalization of the standard block model~\cite{rohe2012co}, where there are $k_1$ blocks of nodes on the left side of the graph and $k_2$ blocks of nodes on the right. The interactions between nodes at two sides are governed by a matrix $B \in \mathbb{R}^{k_1\times k_2}$, i.e., for any $u$ in the $i$-th left block and $v$ in the $j$-th right block, $\Pr[\{u, v\} \in E] = B_{i,j}$. In our setting, each right node corresponds to a word, and words in the same block can be interpreted as in the same topic. Each node to the left corresponds to a sentence, and sentences in the same block corresponds to the same segment. A word is connected to a sentence if it appears in the sentence at least once.}


\qq{Spectral clustering can be used to recover the block structure in an SCBM. This algorithm first finds the leading singular vectors and values of the bipartite graph’s adjacency matrix, and then runs standard clustering algorithms (e.g., k-means) on the leading singular vectors/values. The algorithm is known to be effective for sparse SCBM because it can effectively remove singular vectors along directions that contain stronger noise than signal. A synthetic example of this is shown in Fig.~\ref{sythetic}, where we have generated 2000 sentences and 5000 words according to an SCBM with four blocks each. The dataset is too sparse for standard degree-based algorithms, such as counting shared neighbors~\cite{jarvis1973clustering} or using Jaccard similarity~\cite{niwattanakul2013using}, to recover these blocks (Fig.~\ref{raw_cluster}). Nevertheless, we notice that all the singular vectors/values beyond the first three correspond to noise (Fig. ~\ref{eigenvalues}). Therefore, when we perform clustering only on the leading dimensions of the singular vectors, we are able to exactly recover the blocks (Fig.~\ref{U_cluster}). This motivates us using the singular vectors to cluster words and sentences.
}

\begin{figure}[t]
\vspace{-1.5mm}
\centering
\subfloat[\qq{Top-10 singular values of the sentence-word matrix.}]{
  \centering
  \includegraphics[width=.26\linewidth]{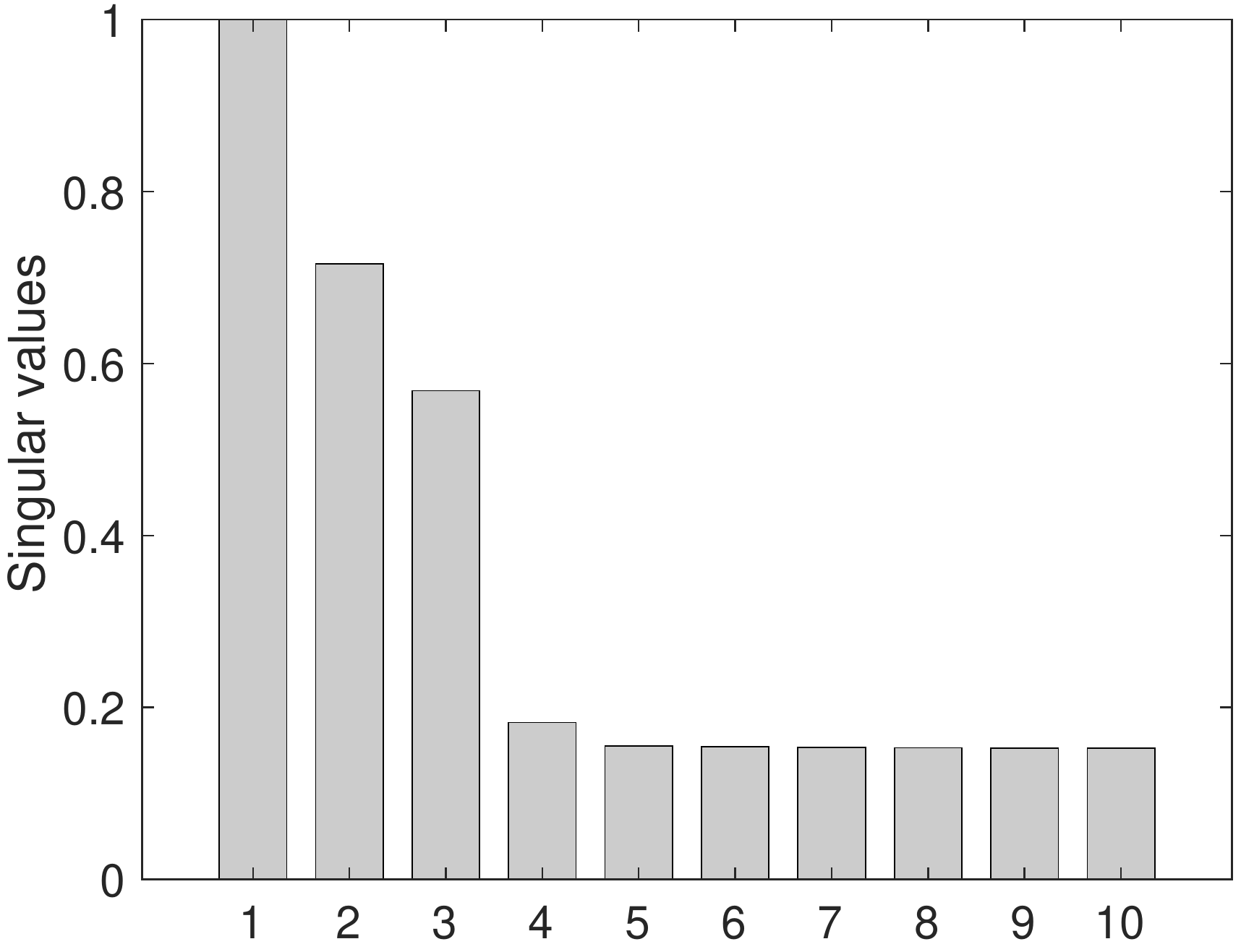}
  \label{eigenvalues}
} \hspace{5mm}
\subfloat[\qq{Heatmap of pairwise distances between words.}]{
  \centering
  \includegraphics[width=.26\linewidth]{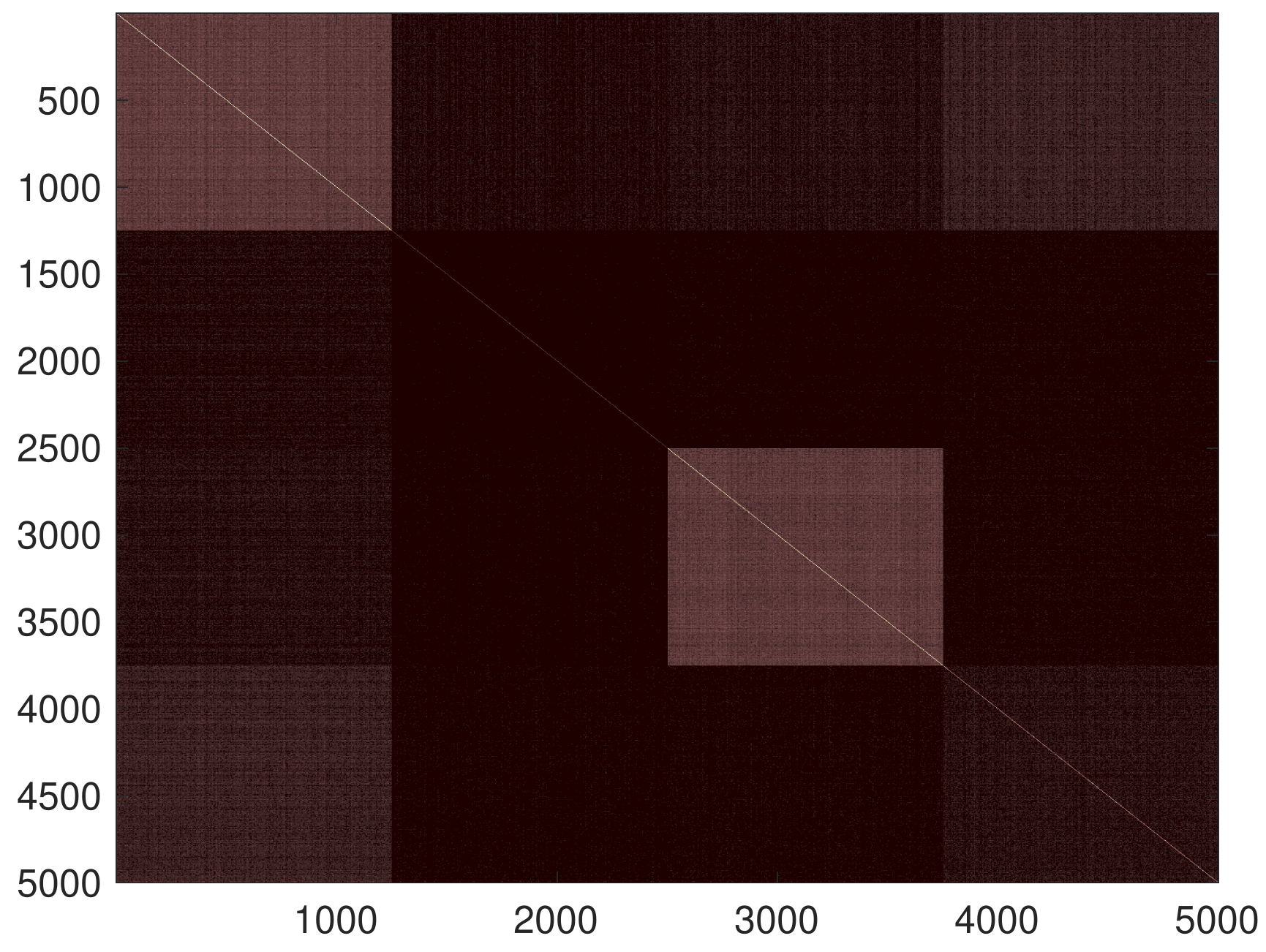}
  \label{raw_cluster}
} \hspace{5mm}
\subfloat[\qq{Right singular vectors plotted in the first two dimensions.}]{
  \centering
  \includegraphics[width=.26\linewidth]{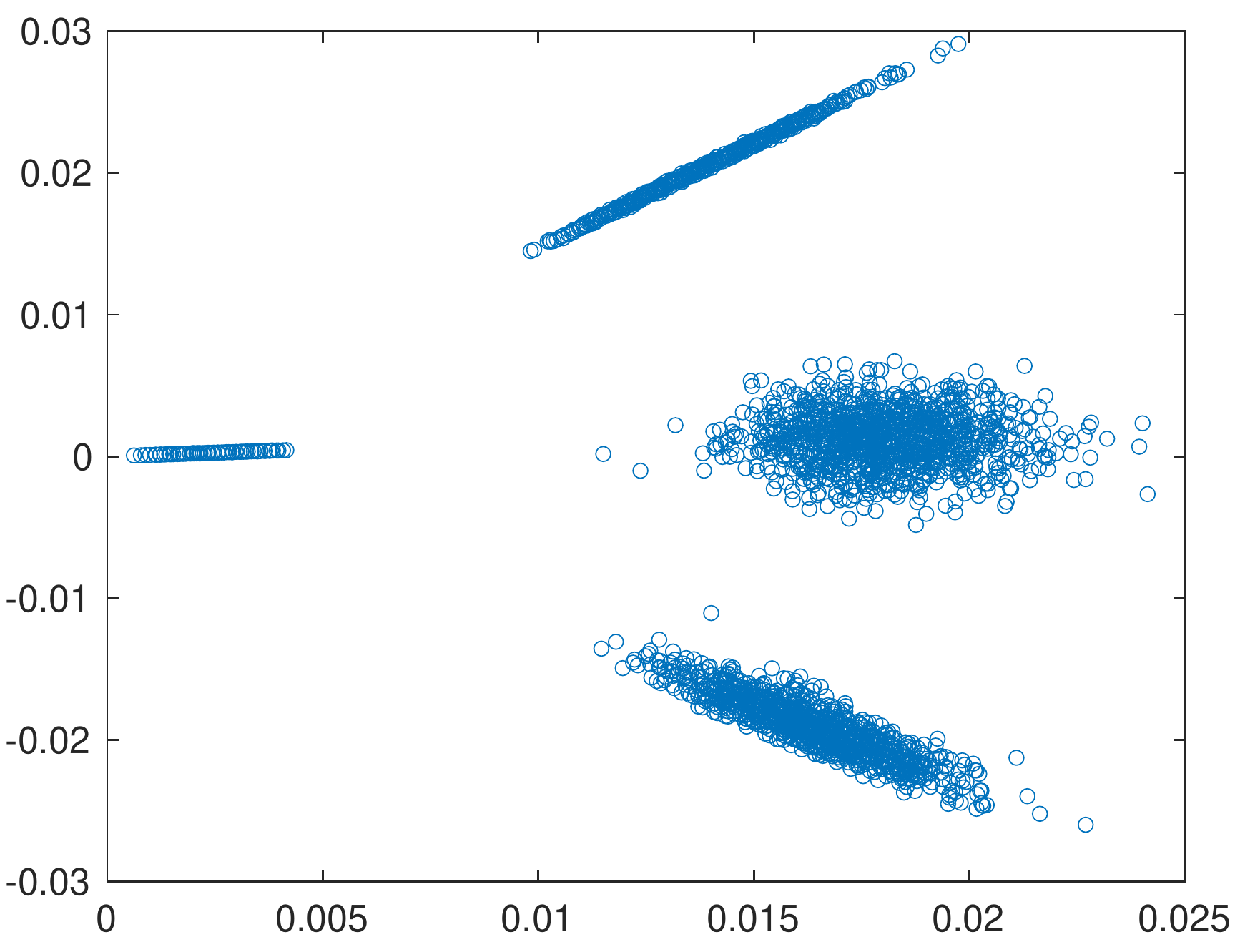}
  \label{U_cluster}
}%
\vspace{-1.5mm}
\caption{\qq{Synthetic example illustrating the use of spectral decomposition to recover the block structure in a sentence-word matrix generated according to an SCBM with four blocks of sentences and words.}}
\label{sythetic}
\end{figure}

\vspace{-2mm}
\subsubsection{Obtaining a low dimensional embedding}
We consider the singular value decomposition (SVD) of the graph Laplacian $L$. By definition, the SVD yields
\begin{equation}
L = U \Sigma V^T,
\end{equation}
where $U \in \mathbb{R}^{m \times m}$ and $V \in \mathbb{R}^{n \times n}$ are unitary matrices and $\Sigma$ contains the singular values $\sigma_1,...,\sigma_{\max \{m,n\}}$ on its diagonal. Since $L^T L = V (\Sigma^T \Sigma) V^T$ is a measure of similarity between words, counting their degrees of connectivity via sentences, and $L L^T = U (\Sigma \Sigma^T) U^T$ is a measure of similarity between sentences, counting their degree of connectivity via words, the SVD can be used to cluster words (using $V$) and sentences (using $U$). Further, as the eigenvalues of $L^T L$ and $L L^T$ are the square of the singular values in $\Sigma$, we introduce another parameter $k$ which denotes the number of left $U_1, ..., U_k \in \mathbb{R}^{m}$ and right $V_1, ..., V_k \in \mathbb{R}^{n}$ dominant singular vectors used, where we assume the singular values are in decreasing order $\sigma_1 \geq \sigma_2 \geq \cdots$. We then re-normalize the rows of the resulting matrices
\begin{equation}
\label{eq:normSVD}
V' = [V'_{i\ell}] = [V'_1 \; V'_2 \; \cdots \; V'_k], \quad U' = [U'_{j\ell}] = [U'_1 \; U'_2 \; \cdots \; U'_k]
\end{equation}
to have unit length, i.e., so that $\sum_\ell V_{i\ell}^{'2} = \sum_\ell U_{j\ell}^{'2} = 1$ for each sentence $i$ and word $j$. Following \cite{von2007tutorial}, which suggests that the dimensionality should be consistent with the number of clusters to be grouped, we use the same parameter $k$ for both $U$ and $V$.


The full decomposition process developed in Sections~\ref{ssec:preprocessing}-\ref{ssec:singular_vectors} is summarized in Algorithm \ref{algo1}.

\vspace{-2mm}
\subsubsection{Impact of $w$ and $d$}
\label{sssec:SVD_heatmap}
Recall the window $w$ and decay $d$ parameters from (\ref{eq:sentence_bonding}). We investigate the impact of these parameters on the matrix decomposition in (\ref{eq:normSVD}) by considering the L2-norm distances between the resulting sentence vectors in $U$. Figure \ref{fig:wd_svd_heatmap} gives heatmaps of these distances for two arbitrary documents in one of our datasets (see Section \ref{ssec:datasets}). Since neighboring sentences should cover similar topics, we seek values of $w$ and $d$ for which ordering information is clearly embedded in the matrix. In Figure \ref{fig:wd_svd_heatmap}(a), for small values of $w$ (i.e., $w = 0, 1$), the sentence order is less clear as the elements near the diagonal are more blurry. As $w$ increases, the pattern becomes more obvious, and when $w = 3$ we observe clear block patterns in the heatmap. When $w$ is increased further (i.e., to $w = 5$), the sharpness of the block pattern does not continue to improve; intuitively, sentences at the far ends of the bonding window for large $w$ will have higher dissimilarity, but this effect is blunted by the decay $d$ (which is $0.7$ here). Since a higher $w$ also increases the runtime of the method, in considering several documents, we find that the best choice of $w$ is typically between $3$ and $5$ (i.e., the number of topic-neighboring sentences is $6$ to $10$).

\begin{figure*}[t]
    \vspace{-4mm}
    \centering
    \includegraphics[scale=0.19]{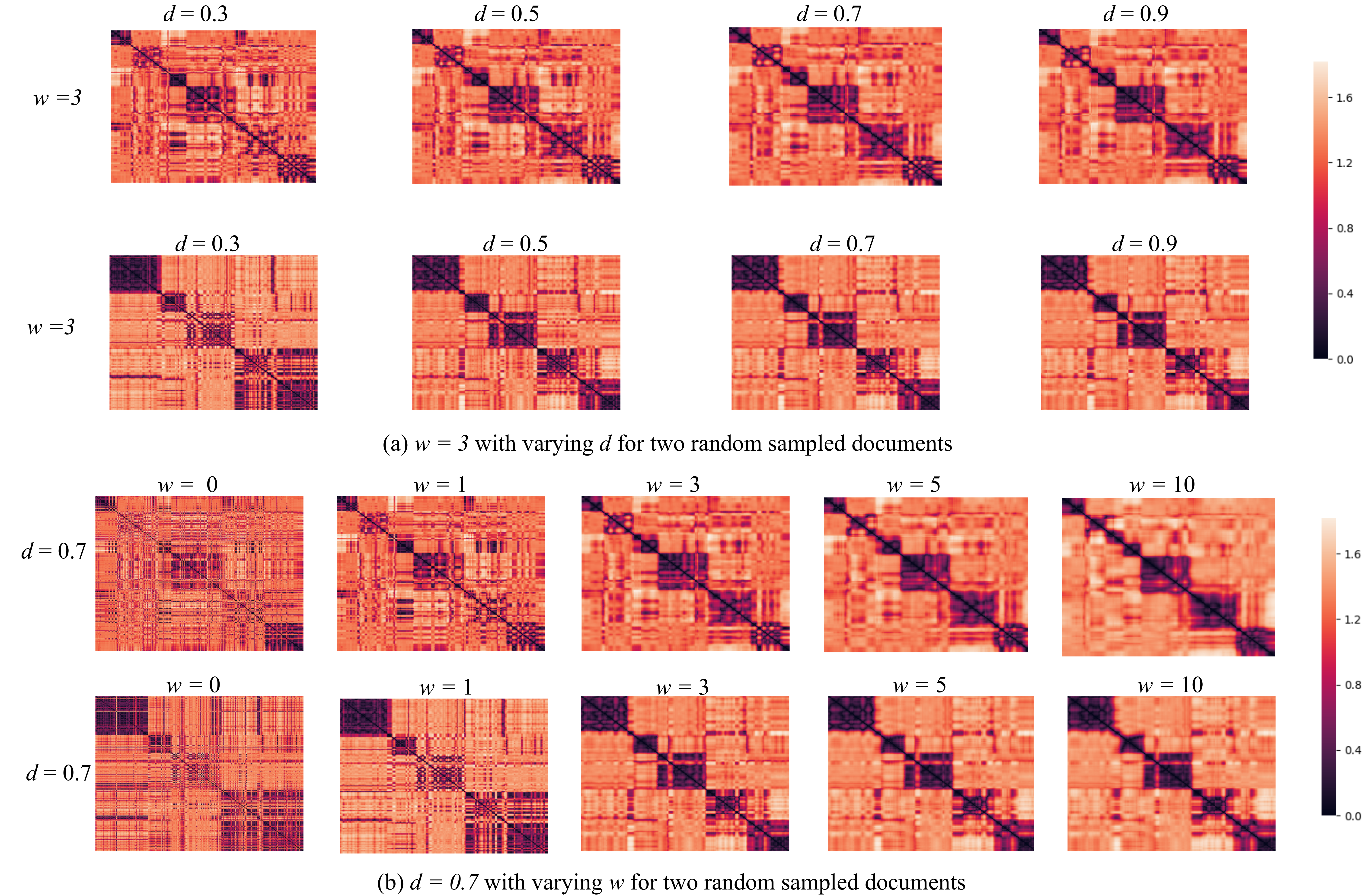}
    \vspace{-2mm}
    \caption{\small{Heatmaps of the pairwise distances between sentence vectors in the SVD for two random sampled document taken from the Introductions dataset  under different values of parameters $w$ and $d$. (a) varies $w$ for fixed $d = 0.7$ and varies $d$ for fixed $w = 3$ on the first randomly sampled document. (b) varies $w$ for fixed $d = 0.7$ and varies $d$ for fixed $w = 3$ on the second randomly sampled document.}}
    \label{fig:wd_svd_heatmap}
\end{figure*}

By this logic, then, the value of $d$ should be significantly lower than $1$. As it is decreased in Figure \ref{fig:wd_svd_heatmap}(b) (i.e., from $d = 0.9$), we see that the sharpness of the blocks improves, with $d = 0.7$ giving the clearest pattern. Beyond this (i.e, $d = 0.5, 0.3$), however, the sharpness begins to decrease. In these cases, neighboring sentences are assigned lower weights, confirming that the SVD uncovers topic similarity between neighbors. In considering several documents, we find that the best choice is $d \approx 0.7$ for this reason. In Section~\ref{sec:experiment}, we will verify that $w = 5, d = 0.7$ leads to the highest average performance on the topic modeling and text segmentation metrics for all datasets considered.  

\vspace{-2mm}
\subsection{Word and Sentence Clustering}
\label{ssec:bicluster}
\vspace{-.6mm}
With the embedding from (\ref{eq:normSVD}) in hand, we move to obtain topics and segments via spectral clustering of the word $V$ and sentence $U$ matrices respectively. Following~\cite{von2007tutorial}, we consider the problem from a graph cut point of view, where the cuts are taken on a similarity graph of words or sentences.

Formally, let $G = [g_{i,j}] \in \mathbb{R}^{n \times n}$ be the similarity matrix among a set of $n$ nodes $v_1,...,v_n$ (i.e., words or sentences), where $g_{i,j} \geq 0$ is the similarity between nodes $v_i$ and $v_j$. We seek to minimize $KCut(S_1, ..., S_k) = \frac{1}{2}\sum_{p=1}^k cut(S_p, \Bar{S_p})$ while $cut(S_p, \Bar{S_p}) = \sum_{\substack{i \in S_p, j \in \Bar{S_p}}}g_{i,j}$. Note $S = \{S_1, S_2, ..., S_k\}$ is a grouping of the nodes into $k$ disjoint sets $S_1,...,S_k$. A simple and straightforward solution for this minimization problem is to cut off individual nodes which are weakly connected to the rest. However, there is usually no topic with one word or text segment with one sentence, therefore, the groups of words or sentences are supposed to have more balanced sizes. As a result, the objective function needs to take group sizes into consideration and build the balanced cut problem
\begin{equation} \label{eq_bcut}
BCut(S_1, ..., S_k) = \sum_{p=1}^k \frac{cut(S_p, \Bar{S_p})}{|S_p|}.
\end{equation}
Taking group sizes into account makes the problem NP hard and requires further relaxation. We reorganise the problem by defining a group indication matrix $H = [h_1 \cdots h_k] \in \mathbb{R}^{n \times k}$ consisting of $k$ weighted indicator vectors $h_p = (h_{1,p}, ..., h_{n,p})^T, p = 1, ..., k$ where
\begin{equation} \label{eq_H}
    h_{i,p} = 
    \begin{cases}
    1 / \sqrt{|S_p|} & \text{if node $v_i \in S_p$}\\
    0 & \text{otherwise.}
    \end{cases}
\end{equation}
We can see $H^T H = I$ where $I$ is the identity. Letting $G$ be the node similarity graph, we define its degree matrix as $D = \mbox{diag}(d_1,...,d_n)$ where $d_i = \sum_j g_{ij}, i=1,...,n$, and the unnormalized graph Laplacian as $L_u = D - G$. Through some easy math, we can get
\begin{equation} \label{eq_hlh}
    h_p^T L_u h_p = \frac{cut(S_p, \Bar{S_p})}{|S_p|}, \text{ for } p = 1,...,k.
    \vspace{-1mm}
\end{equation}
Combining this with (\ref{eq_bcut}), we conclude that
\begin{equation} \label{eq_trHLH}
\vspace{-1mm}
BCut(S_1, ..., S_k) = \sum_{p=1}^k h_p^T L_u h_p = Tr(H^T L_u H).
\vspace{-1mm}
\end{equation}
Thus, the minimization problem can be presented as
\begin{equation} \label{eq_conMin}
\begin{split}
    \min_{\substack{S_1, \cdots, S_k}} Tr(H^TL_uH) \text{ subject to $H^TH=I$, $H$ defined as (\ref{eq_H})}. 
\end{split}
\end{equation}
This problem is equivalent to minimizing (\ref{eq_bcut}) and is NP-hard. Therefore, in the BATS methodology, we relax this constraint by allowing $h_{i,p} \in \mathbb{R}$ to take any arbitrary value, and turn \eqref{eq_conMin} into 
\begin{equation} \label{eq_relMin}
    \min_{\substack{H \in \mathbb{R}^{n\times k}}} Tr(H^TLH)\text{ subject to } H^TH=I. 
    \vspace{-1mm}
\end{equation}
This approach allows us to employ clustering algorithms to solve the minimization problem. In the following sections, we detail our methods for solving (\ref{eq_relMin}) to cluster words (Section \ref{ssec:topics}) and sentences (Section \ref{ssec:segments}), respectively. 

\vspace{-2mm}
\subsubsection{Topics via word clustering}
\label{ssec:topics}
To obtain the topics, we consider spectral clustering for the normalized word matrix $V'$ in (\ref{eq:normSVD}). Since each row $v'_j \in \mathbb{R}^k, j = 1,...,n$ of $V$ is a $k$-dimensional representation of a word, the clustering optimization in (\ref{eq_trHLH}) takes these $n$ words as the nodes to be grouped into $k$ sets $S_1, S_2, ..., S_k$ based on the similarities $g_{j,j'}$ between pairs of word representation vectors $v_j$ and $v_{j'}$. This is equivalent to minimizing the pairwise deviations between representations of nodes within the sets:
\begin{equation} \label{eq_kmeans}
    S = \underset{\{S_1,...,S_k\}}{\arg \min} \sum_{p = 1}^{k} \frac{\sum_{j, j' \in S_p} ||v'_{j} - v'_{j'}||^2}{2|S_p|}.
\end{equation}
This optimization is equivalent to a KMeans clustering \cite{lloyd1982least} of the vectors $v'_1, ..., v'_n$. The number of clusters is determined by the number of segments $k$, and each resulting word cluster $S_p$ refers to a topic. To obtain a description of each topic in terms of its top words, we further rank the words in each cluster according to the standard term frequency-inverse document frequency (tf-idf) metric \cite{rajaraman2011mining} applied to the awarded sentence-word matrix $X^a$ in (\ref{eq:pos_awarding}). The tf-idf assignment matrix $F$ is obtained during the matrix decomposition process and it is the same size of $X^a$. To assign each word a single tf-idf score for sorting, we sum the tf-idf scores of each word over all sentences. 

\vspace{-2mm}
\subsubsection{Segments via sentence clustering}
\label{ssec:segments}
We then turn to clustering the normalized sentence matrix $U'$ in (\ref{eq:normSVD}) to obtain the segments. Compared with the topic clustering problem, this one will have more constraints since the clusters are related to the sentence orders and the cluster sizes can be uneven. As a result, the KMeans method is no longer applicable, and we resort instead to an agglomerative clustering method with connectivity constraints \cite{davidson2005agglomerative} to solve (\ref{eq_trHLH}). In agglomerative clustering, nodes are grouped together sequentially according to pairwise similarities: the process recursively merges two groups of nodes that yield the minimum between-cluster distance. 

Formally, recall that the sentence embeddings are the rows $u'_i \in \mathbb{R}^k, i = 1,...,m$ of the matrix $U'$. We form the graph of sentences $G_S = (V, E_S)$, where $V = \{i | i = 1, 2, ..., m\}$ and $E_S = \{(i, j) | i, j = 1, 2, ..., m, i \neq j \}$, with the weight of the edge $(i, j) \in E_S$ being the cosine similarity between $u'_i$ and $u'_j$. The ordering constraint should be such that only adjacent sentences can be clustered; we therefore initialize a connectivity graph $G_C^1 = (V, E_C^1)$ where for all pairs of nodes $i, j \in V$, $(i, j) \in E_C^1$ if and only if $j = i + 1$, i.e., each node connects to the next sentence. Letting $S_i^r$ denote cluster $i$ of the sentences at the $r$th iteration, initialized as $S_i^1 = \{i\}$ for each $i$, the merging operation of our constrained agglomerative clustering is given by
\begin{equation} \label{eq_agg}
         (S_{p}^{r}, S_{q}^{r}) = \underset{(i, j) \in E_C^r}{\arg\min} \; \overline{D}(S_{i}^{r}, S_{j}^{r}), \quad
         S_p^{r+1} = S_{p}^{r} \cup S_{q}^{r}, \quad
         S_q^{r+1} = \emptyset, \quad
         E_C^{r+1} = E_C^r \setminus \{(p,q)\} \cup \{(p, a^r(q)\},
\end{equation}
for $r=1,...,m-k$, where $\overline{D}(S_{i}^{r}, S_{j}^{r})$ refers to the distance between the sets $S_i^r$ and $S_j^r$, which is treated as the average distance between sentences in $S_i^r$ and $S_j^r$ according to their link weights in $E_S$, and $a^r(q)$ is the single node that $q$ points to in $G_C^r$. In each iteration, the two adjacent clusters of sentences that have minimum distance are merged together. The procedure ends after $r = m-k$ iterations, when there are $k$ clusters $i$ for which $S_i^k \neq \emptyset$; these are taken as the segments.

\vspace{-2mm}
\subsubsection{Advantages of joint topic modeling and text segmentation}
\label{sssec:joint_advantages}
In the BATS methodology, segmentations are produced from sentence clusters whereas topics are produced from word clusters. A key observation we have made in the design of BATS is that we can indeed use one simple bi-clustering algorithm to simultaneously identify sentence and word clusters, based on two simple assumptions: (i) when two sentences are in the same cluster (segment), the same set of words are more likely to appear in both sentences, and (ii) when two words are in the same cluster (topic), these words are more likely to co-appear in the same sentences.

\begin{algorithm}[t]
\small
\setstretch{1.2}
\caption{BATS topic modeling and text segmentation.}
\bgroup
\def\arraystretch{1.1}
\begin{algorithmic}[1]
    \Statex \textbf{INPUT:} Single text document, segment number $k$
    \Statex \textbf{PARAMETER:} Awarding value $\lambda$, window size $w$, decaying rate $d$
    \Statex \textbf{OUTPUT:} Topic words, text segments
    \Procedure{MainProcess}{text, $k$, $\lambda$, $w$, $d$}
    \State Remove degree-one words from text.
    \State Compute sentence-word matrix $X^o$ and POS-based matrix $T$.
    \State $U', V', F \gets$ MAT\_DECOMP($X^o$, $T$, $\lambda$, $w$, $d$, $k$).  // Alg.\ref{algo1}
    \State \multiline{%
    Cluster the rows of $V'$ with KMeans into $k$ clusters. Sort the words in each cluster by $F$ (tf-idf scores).}
    \State \multiline{%
    Cluster the rows of $U'$ with agglomerative clustering into $k$ clusters with a connectivity constraint.}
    \State Topic words $\gets$ Sorted words in each cluster
    \State Text segments $\gets$ Sentence clusters
    \State \textbf{return} Topic words, Text segments
    \EndProcedure
\end{algorithmic}
\egroup
\label{algo2}
\end{algorithm}
\setlength{\intextsep}{0pt}

In other words, the segmentation and topic modeling tasks are coupled through the bi-clustering algorithm in BATS, which uses sentence-word interactions to identify a joint latent structure for segments and topics. The model of which sentences are contained in which segments simplifies the process of identifying topics, and vice versa. Let us consider a concrete example, in which a document consists of 10 sentences. Each sentence may have different topic distributions, e.g., sentence 1 has 80\% of words from a ``science'' topic and 20\% from a ``health'' topic, whereas sentence 7 has 40\% from the ``science'' topic and 60\% from the ``health'' topic. If we assume that words from the same segment come from the same topic distribution, when the ground-truth of the segmentation is known -- e.g., segment 1 consists of the first four sentences and segment 2 consists of the last six sentences -- we are narrowing the set of topics these words are assumed to be sampled from. Specifically, in this example, we effectively observe two sets of words sampled from two topic distributions (one for each segment), instead of 10 (one for each sentence). This extra information simplifies the process of identifying topics. In fact, the bi-clustering algorithm implicitly uses this information in the singular value decomposition step: it uses sentence clusters (segments) to improve word clusters, and vice versa. Therefore, we expect that doing topic modeling and text segmentation jointly will improve the quality of both tasks, with the added advantage of being more computationally efficient than doing both separately.

\vspace{-2mm}
\subsection{BATS Methodology Summary}
\vspace{-.6mm}
The full BATS topic modeling and text segmentation methodology (including Algorithm \ref{algo1}) developed throughout this section is summarized in Algorithm \ref{algo2}. The inputs are the single text document of interest and $k$, the number of topics and segments to extract. The algorithm begins with denoising, which removes all degree-one words, and constructing the sentence-word matrix $X^o$ and parts-of-speech matrix $T$. $X^o$ and $T$ are then inputted to the matrix decomposition procedure, detailed in Algorithm \ref{algo1}, which employs sentence bonding and graph Laplacian regularization to obtain the matrices $U'$ and $V'$, containing the encodings of the sentences and words, and the tf-idf matrix $F$. The rows of $V'$ are then clustered into $k$ clusters of words via KMeans, with the words in each cluster sorted by tf-idf score in $F$, forming the topics. Finally, the rows of $U'$ are clustered into $k$ clusters of sentences via constrained agglomerative clustering, forming the segments.

\vspace{-2mm}
\subsection{Time Complexity Analysis}
\label{ssec:complexity}
\vspace{-.6mm}
Recall from Section~\ref{ssec:motivation} that an important setting for the ``single and new document'' setting is in the presence of constrained computational resources. Thus, we also perform a complexity analysis to investigate the efficiency of our algorithm, given the importance of low runtimes.

From Algorithm \ref{algo1}, note that there are three main procedures in BATS which have major impacts on the time complexity: matrix decomposition, KMeans clustering for words, and agglomerative clustering for sentences. The matrix decomposition process consists of multiple matrix multiplication and summation operations, of which matrix multiplication dominates with a complexity of $O(\max(m^3, n^3))$, where $m$ and $n$ are the number of sentences and words, respectively. However, in our application, the sentence-word matrix is sparse and therefore enables some optimizations in the matrix decomposition procedure. Specifically, iterative methods such as primme~\cite{stathopoulos2010primme} can decompose sparse matrices with a time complexity of $O(mnr)$, where $r$ is the number of singular vectors. Formally, we can show that when the number of singular vectors is small, the time complexity is described as follows:
\begin{lemma}
\vspace{-2mm}
\label{lemma1}
For a given document, assume $m$ and $n$ are the number of sentences and words, respectively. If the number of non-zero singular values of the sentence-word matrix $X^a$ is sufficiently small, then the runtime of BATS on the document can be approximated as $O(mnr +n)$.
\vspace{-2mm}
\end{lemma}
\begin{proof}
In the KMeans clustering procedure, all $n$ word vectors are compared to $k$ centroids to find the closest centroid, and this step iterates $t_K$ times, leading to the time complexity of $O(nt_Kk)$. In the constrained agglomerative procedure, the similarities between $m$ sentence vectors are computed for clustering, and with $t_A$ iterations, the time complexity is $O(t_A m \log m)$ with the efficient priority queue implementation \cite{manning2008introduction}. For the matrix decomposition, as discussed above, using an iterative method such as primme leads to a time complexity of $O(mnr)$. The overall time complexity of our method is the sum of all these procedures, which leads to
\vspace{-2mm}
\begin{equation} \label{eq_comp}
    O\left(mnr + n t_K k + t_A m \log m\right).
    \vspace{-2mm}
\end{equation}

Since $O(mnr) \gg O(m \log m)$, $k \ll n$, and $r \ll m$ by assumption of a low-rank sentence-word matrix, this complexity can be approximated as $O(mnr + n)$.
\vspace{-3mm}
\end{proof}

Thus, BATS has a much lower time complexity than $O(\max(n^3, m^3))$, with the difference being particularly pronnounced when $n \gg m$. The scalability and small runtimes of BATS will be verified experimentally in Section \ref{ssec:scala}.

\section{Experimental Evaluation and Discussion}\label{sec:experiment}
\vspace{-.6mm}
We turn now to evaluating our BATS methodology. After describing the datasets (Section \ref{ssec:datasets}), we consider performance against baselines on the topic modeling (Section \ref{ssec:res-topic}) and text segmentation (Section \ref{ssec:res-seg}) tasks. Then, we consider the scalability of our method (Section \ref{ssec:scala}). Finally, we conduct an ablation study to determine the impact of various components of BATS (Section \ref{ssec:ablation}). 

\vspace{-2mm}
\subsection{Description of Datasets}
\label{ssec:datasets}
\vspace{-.6mm}
We consider documents from \qq{six} datasets -- Textbook, Lectures, Introductions, Choi, \qq{Wiki, and News} -- obtained from different text applications. Basic statistics on these datasets are given in Table \ref{tabD}, including the number of documents, the average sentences per document, the average word counts per document, and the average sparsity per document (fraction of zero entries in the $X^a$ matrix), both before and after the preprocessing procedures described in Section \ref{ssec:preprocessing}. More specifics on these datasets are as follows:

\noindent \textit{(i) Textbook dataset:} This is drawn from the medical textbook in \cite{walker1990oral}. Each chapter is treated as a document, and each section as a segment. The numbers of segments per document and sentences per segment have high variance. Moreover, segments within a document tend to be similar in their constituent words, as they are different sections of the same chapter. As a result, this dataset helps us test on cases where documents have different segments discussing similar topics.

\bgroup
\begin{table*}[t]
\vspace{-4mm}
\centering
\captionsetup{font=small}
\caption{Basic statistics of the \qq{six} datasets used for evaluation.}
\vspace{-4mm}
\def\arraystretch{1.1}
\small
\begin{center}
\resizebox{\linewidth}{!}{\begin{tabular}{l|c|c|c|c|c|c|c|c|c}
    \hline
    \hline
    Dataset & Documents & Avg. sentences & Avg. segments & Avg. words & Avg. words & Avg distinct words & Avg. distinct words & Avg. sparsity & Avg. sparsity \\
        & & per doc & per doc & before preproc. & after preproc. & before preproc. & after preproc. & before preproc. & after preproc. \\
    \hline
    Textbook & 227 & 136 & 4 & 3551 &  2590 & 640 & 241 & 97.9\% & 75.2\%\\
    Lectures & 55 & 392 & 8 & 8002 &  4924 & 686 & 342 & 99.0\% & 89.0\%\\
    Introductions & 2135 & 195 & 5 & 7022 &  4752 & 1016 & 449 & 98.6\% & 84.8\%\\
    Choi & 920 & 74 & 10 (const) & 1673   & 1489  & 650 & 162 & 98.0\% & 76.5\% \\
    \qq{Wiki} & 727,000 & 60 & 4 & 1154 & 1007 & 521 & 334 & $96.3\%$ & $89.6\%$\\ 
    \qq{News} & 300,000 & 51 & 4 &  1096& 730 & 516 & 321 & $97.1\%$& $87.2\%$ \\
    \hline
\end{tabular}}
\label{tabD}
\vspace{-1mm}
\end{center}
\end{table*}
\egroup


\noindent \textit{(ii) Lectures dataset:} This dataset contains transcripts of conversational lectures on AI and physics topics~\cite{eisenstein2008bayesian}.
As each lecture is divided into sections, we treat lectures as documents and sections as segments. Each lecture script has 6-10 sections. Compared with the other datasets, the sentences are more conversational, tending to be shorter and simpler. Therefore, this dataset helps us examine algorithm performance on lengthy conversational documents.

\noindent \textit{(iii) Introductions dataset:} In this dataset, every document is an artificial combination of abstracts and introductions from academic articles in different fields~\cite{statistical_models}. We randomly choose 3-8 articles, extract the abstract and introduction as one sample, and combine multiple samples into one document. Each sample is treated as one segment. Compared with the other datasets, this will allow us to test on cases with large segment sizes, uneven segment lengths, and a diverse set of topics. 

\noindent \textit{(iv) Choi dataset:} This is a standard dataset \cite{choi2000advances} widely used to evaluate text segmentation approaches. The documents in the dataset are artificial combinations of the first $\ell$ sentences of the documents in the Brown corpus \cite{francis79browncorpus}. Each document has 10 segments, with few sentences per segment. Because the dataset lacks explicit topic distributions and contains mostly segments that are too short for topic modeling, we use it only for evaluating text segmentation.

\noindent \qq{\textit{(v) Wiki dataset:} This is the WIKI-727k dataset from~\cite{koshorek2018text} which is comprised of over 727,000 articles from English Wikipedia. Its vocabulary size exceeds 800,000 tokens. We treat articles as documents and, following~\cite{koshorek2018text}, use the Punkt Sentence Tokenizer from the \texttt{nltk} library \cite{bird2009natural} to generate ground truth segments. We include this dataset primarily to test the performance of BATS on large corpora, although it also provides a new genre of natural text covering a wide variety of encyclopedia topics.}

\noindent \qq{\textit{(vi) News dataset:} Finally, we include a news dataset generated by following links shared on Twitter. We use a collection of Tweets from \cite{li2017world} that focus on news related to the 2016 US presidential election and follow the shared links to scrape text from the linked articles. The news dataset consists of about 300,000 documents with a vocabulary size of over 80,000. Each article is treated as a document and ground truth segments are generated using the Punkt Sentence Tokenizer.}



\vspace{-2mm}
\subsection{Topic Modeling}
\label{ssec:res-topic}
\vspace{-.6mm}


\subsubsection{Baselines}
\label{ssec:app_topicbase}
We compared BATS against six baselines for topic modeling:

\noindent \textit{(i) Latent Dirichlet Allocation (LDA)} \cite{blei2003latent,sarne2019unsupervised}: LDA is a probabilistic topic model which uses two independent Dirichlet priors for the document-topic and word-topic distributions. It trains a model to best estimate the Bayesian probabilities $P(word|topic)$ and $P(topic|document)$. We use the {\tt sklearn} implementation in Python with the default parameters.

\noindent \textit{(ii) Hierarchical Dirichlet Process (HDP)} \cite{teh2005sharing}: HDP is a mixed-membership model which extends LDA to an unknown number of topics by building a hierarchy. 
Specifically, it builds a two-level hierarchical Dirichlet process at the document-level and the word-level to perform parameter inference. We use the {\tt gensim} implementation in Python with the default parameters.

\begin{figure}[t]
\vspace{-2mm}
\captionsetup{font=small}
\centerline{\includegraphics[scale=0.3]{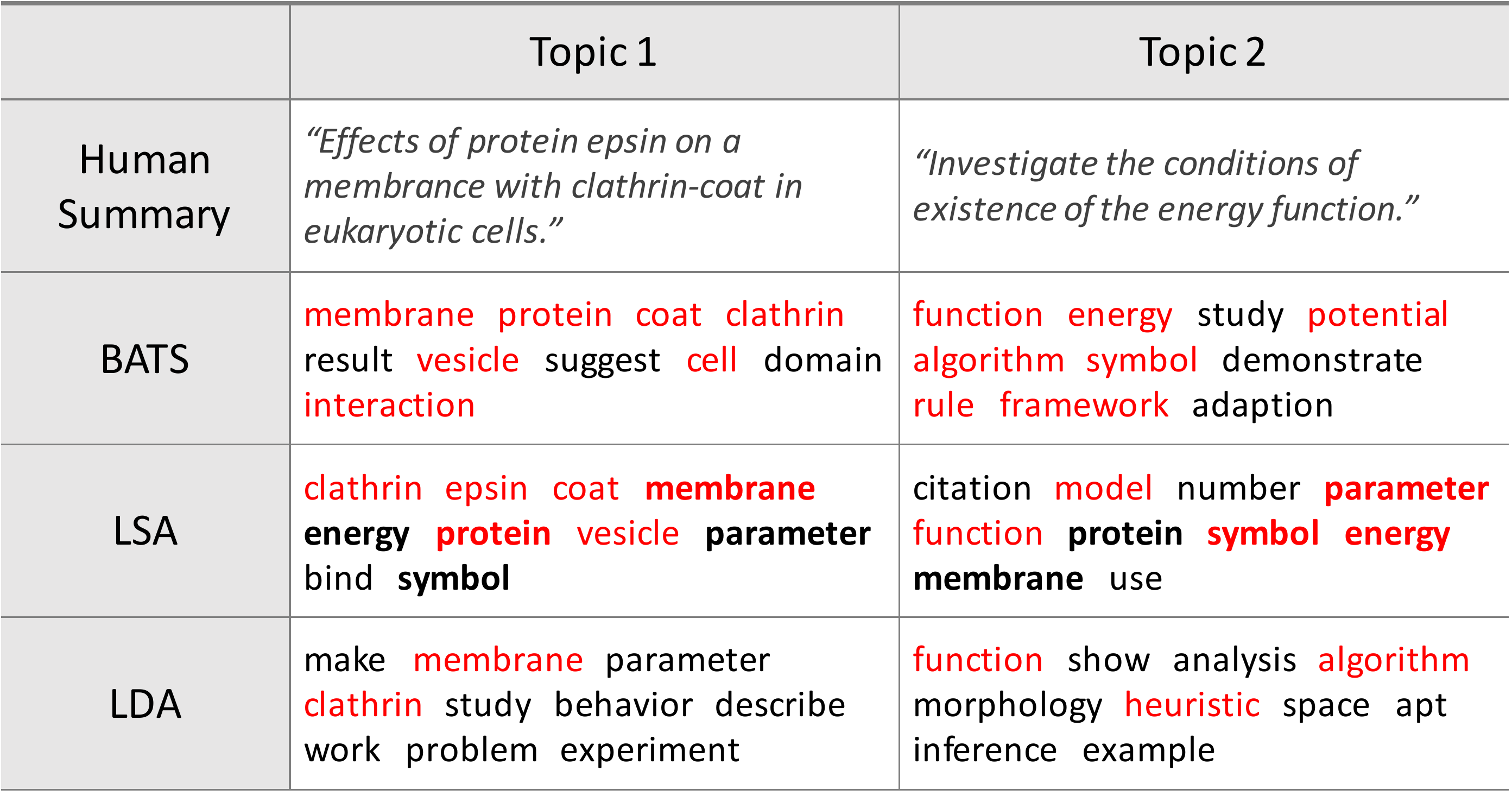}}
\vspace{-0.2cm}
\caption{Example of topics extracted from an arbitrary document in the Introductions dataset. Words in color red are those consistent with a human-generated summary, and duplicated words are boldfaced. Our results produce the best descriptions as well as the least overlaps.}
\label{fig_word}
\vspace{-4mm}
\end{figure}
\begin{figure}[t]
\centerline{\includegraphics[scale=0.35]{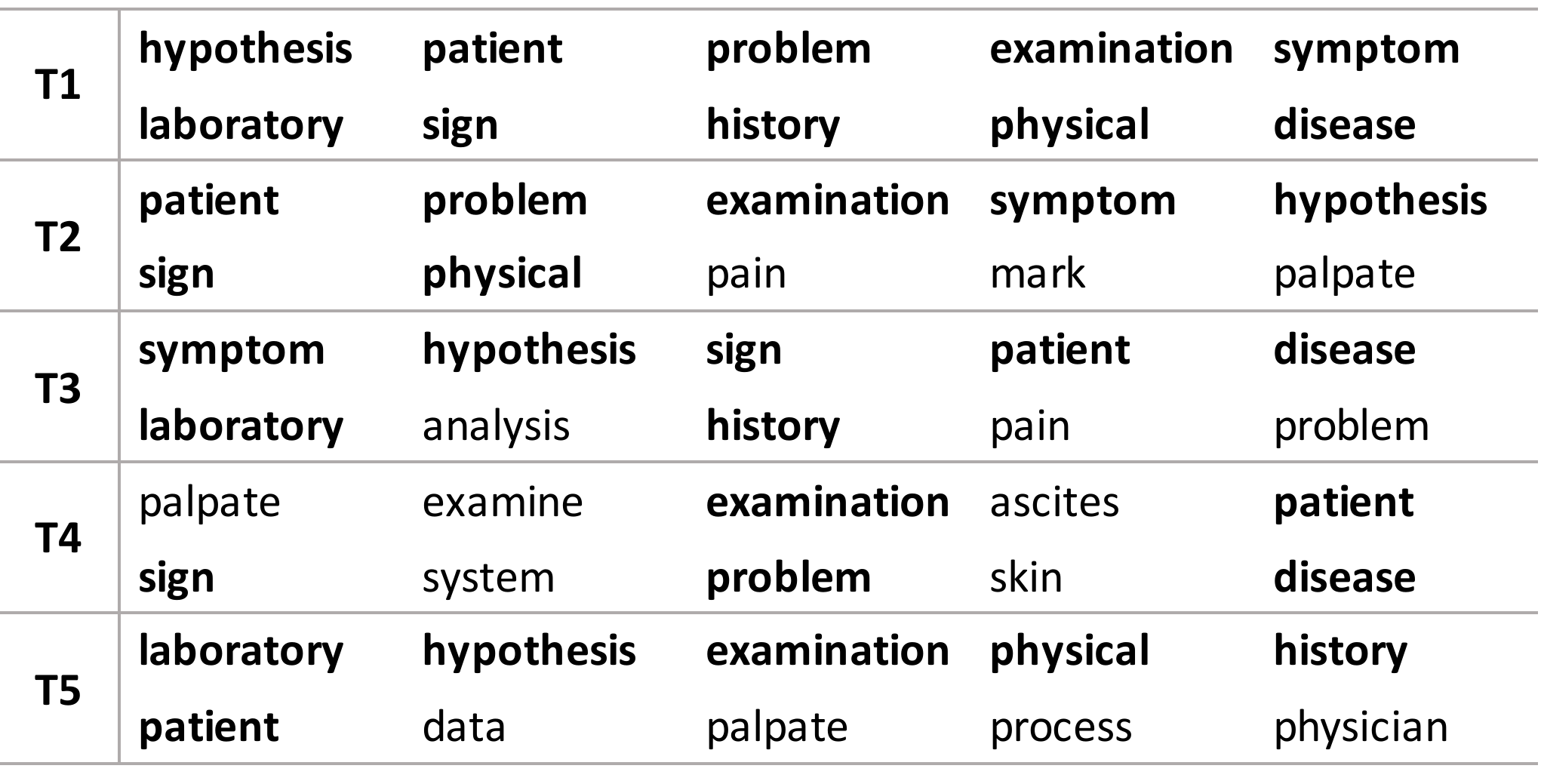}}
\captionsetup{font=small}
\vspace{-0.4cm}
\caption{Example of topics extracted from one document (known to have five topics) in the Textbook dataset by LSA. Duplicated words are denoted in boldface. There is high overlap, motivating the need to consider topic diversity in addition to coherence.}
\label{fig_dup_topic}
\end{figure}

\noindent \textit{(iii) Latent Semantic Analysis (LSA)} \cite{deerwester1990indexing}: LSA decomposes a document-word matrix, based on TF-IDF scores, into a document-topic matrix and a topic-word matrix; the decomposition is performed through a truncated SVD technique. We use the {\tt gensim} implementation in Python.

\noindent \textit{(iv) Probabilistic Latent Semantic Analysis (pLSA)} \cite{hofmann2017probabilistic}: pLSA is developed from LSA, using a probabilistic method instead of SVD to find the latent topics via generative modeling of the observed document-word matrix. We implement pLSA de-novo in Python, using 30 as the max number of iterations, 10.0 as the breaking threshold, and $k$ as the number of topics.

\noindent \textit{(v) Non-negative Matrix Factorization (NMF)} \cite{pauca2004text}: NMF is a linear-algebraic model which factorizes a high-dimensional matrix into two lower-dimensional ones. In this  case, NMF decomposes the document-word matrix (based on TF-IDF scores) into a topic matrix and a coefficient matrix for the topics. We use the {\tt sklearn} implementation in Python with the default parameters.

\noindent \textit{(vi) Semantics-assisted Non-negative Matrix Factorization (SeaNMF)} \cite{shi2018short}: SeaNMF introduces word-context semantic correlations into NMF to extract topics particularly from short texts. The semantic correlations between the words and their contexts are learned from the Skip-Gram with Negative Sampling (SGNS) word embedding technique \cite{mikolov2013efficient, mikolov2013distributed} to address sparsity. The objective function of SeaNMF preserves both the word-document matrix and semantic correlation matrix. We use the author's Python implementation available at \url{https://github.com/tshi04/SeaNMF}.

Since our focus is on single document topic modeling, we evaluate the models on each document separately. Given that the baselines usually learn across multiple documents, to provide a fair comparison, we treat the sentences within each document as the ``documents'' for the baselines, i.e., we feed them the preprocessed sentence-word matrices. For each document, the number of topics assumed by each baseline is taken to be the number of segments. The performance of each baseline is averaged over several trials.

\vspace{-2mm}
\subsubsection{Evaluation metrics}
We employ two popular coherence metrics  to  assess  extracted  topic: pointwise mutual information (PMI) \cite{fano1961transmission} and UMass \cite{mimno2011optimizing}. Higher values of these metrics have been associated with better performance in terms of interpretability and consistency of topics with human evaluation \cite{roder2015exploring}. Since these metrics treat topics separately, in order to evaluate the diversity between topics, we also include two similarity measures: Jaccard (Jacc) Index and S\o rensen-Dice (Dice) Index \cite{gomaa2013survey}. They measure overlaps in words between the topics, with lower values (i.e., less overlap) being better. More specifically:

\begin{table}[t!]
\centering
\vspace{-2mm}
\captionsetup{font=small}
\caption{The PMI values with the varying number of words in each topic ($\mathcal{K}$) on three datasets.}
\vspace{-3.5mm}
\begin{adjustbox}{width=0.6\textwidth}
\begin{tabular}{l|l|l|l}
\hline
$\mathcal{K}$   & Textbook Dataset & Lectures Dataset & Introduction Dataset \\ \hline
2  & 2.13             & 1.25             & 0.72                 \\
5  & 1.45             & 0.92             & 0.64                 \\
10 & 0.62             & 0.54             & 0.58                 \\
20 & -0.12            & -0.51            & -0.33                   \\ \hline
\end{tabular}
\end{adjustbox}
\label{tab:changing_K}
\end{table}

\noindent \textit{(i) Topic coherence measures:} We use the PMI score, \qq{an external} metric which takes the top-$\mathcal{K}$ words under each topic into consideration and has been widely used in recent papers for evaluating topic coherence~\cite{hajjem2017combining,li2019integration,nikolenko2016topic,quan2015short,schneider2018topic,shi2018short,yao2017incorporating,levy2014neural}. Formally, for each topic $l$, the PMI score is calculated as
\vspace{-1mm}
\begin{equation}
\label{eq:PMI}
    PMI_l = \frac{2}{\mathcal{K}(\mathcal{K}-1)} \sum_{1 \leq i <j\leq \mathcal{K}}\log \frac{p(w_i, w_j)}{p(w_i)p(w_j)},
    \vspace{-1mm}
\end{equation}
where $\mathcal{K}$ is the number of words from topic $l$ that are considered, $p(w_i,w_j)$ is the number of documents containing both of the words $w_i$ and $w_j$ divided by the number of documents in the dataset, and $p(w_i)$ and $p(w_j)$ are the number of document containing the words $w_i$ and $w_j$ divided by the total number of documents, respectively. When selecting the $\mathcal{K}$ words from topic $l$, if the topic modeling algorithm orders its results (e.g., from most likely to be part of the topic to least likely) then the top $\mathcal{K}$ words are used. The average PMI score over all the topics is then used to evaluate the quality of the topic models. Following~\cite{shi2018short, schneider2018topic}, we set the default value of $\mathcal{K}$ to 10. \qq{We calculate $p(w_i, w_j)$,$p(w_i)$, and $p(w_j)$ over the entire dataset we are evaluating.}

By contrast, UMass is an intrinsic evaluation metric which takes the sequence of words into consideration by computing the conditional log-probability of each pair of words; the pairwise scores are not symmetric, and therefore the order of the words matters. The UMass score is given as 
\begin{equation}
\label{eq:UMass}
    UMass_l = \sum_{j=2}^{\mathcal{K}} \sum_{i = 1}^{j-1}\log \frac{p(w_i, w_j) + 1}{p(w_i)},
    \vspace{-1mm}
\end{equation}
where $\mathcal{K}$, $p(w_i, w_j)$, and $p(w_i)$ have the same meaning as in \eqref{eq:PMI}. Note that this equation assumes the top-$\mathcal{K}$ words have been ordered from most to least likely to be part of topic $l$. In our single-document evaluation, we consider the internal corpus for UMass to be the document itself. \qq{Note that we include both an external (PMI) and an internal (Umass) metric for completeness to ensure that our evaluation does not give an unfair advantage to any particular method.}

\noindent \textit{(ii) Similarity score measures:} With $T_i$ and $T_j$ as the sets of words comprising topics $i$ and $j$, the Jacc $Jacc(T_i, T_j)$ and Dice $Dice(T_i, T_j)$ similarity scores are computed as:
\begin{equation} \label{eq_jc}
    Jacc(T_i, T_j) = \frac{|T_i \cap T_j|}{|T_i \cup T_j|}, \qquad Dice(T_i, T_j) = \frac{2|T_i \cap T_j|}{|T_i| + |T_j|}.
    \vspace{-1mm}
\end{equation}
The example in Figure \ref{fig_word} shows the importance of considering both types of metrics. The topics extracted by the LSA baseline tend to have many duplicated words (50\% in the example) as compared with results from LDA and BATS, even though it has roughly the same number of words that are consistent with a human-generated summary as our method. Further, since the overall scores for each document are averaged across topics, poor results in terms of one metric on any given topic can be outweighed by high performance on other topics. Since the overall topic coherence scores for each document are averaged across topics, similar topics with duplicate words and high coherence scores will achieve a high average score. Figure \ref{fig_dup_topic} shows another example of this for LSA: though this method achieves high topic coherence, the topics are highly overlapped, motivating the need to take diversity into consideration.

\begin{figure}[t]
\centering
\subfloat [Average PMI score on the \\ Introduction dataset.]{
  \includegraphics[width=.30\linewidth]{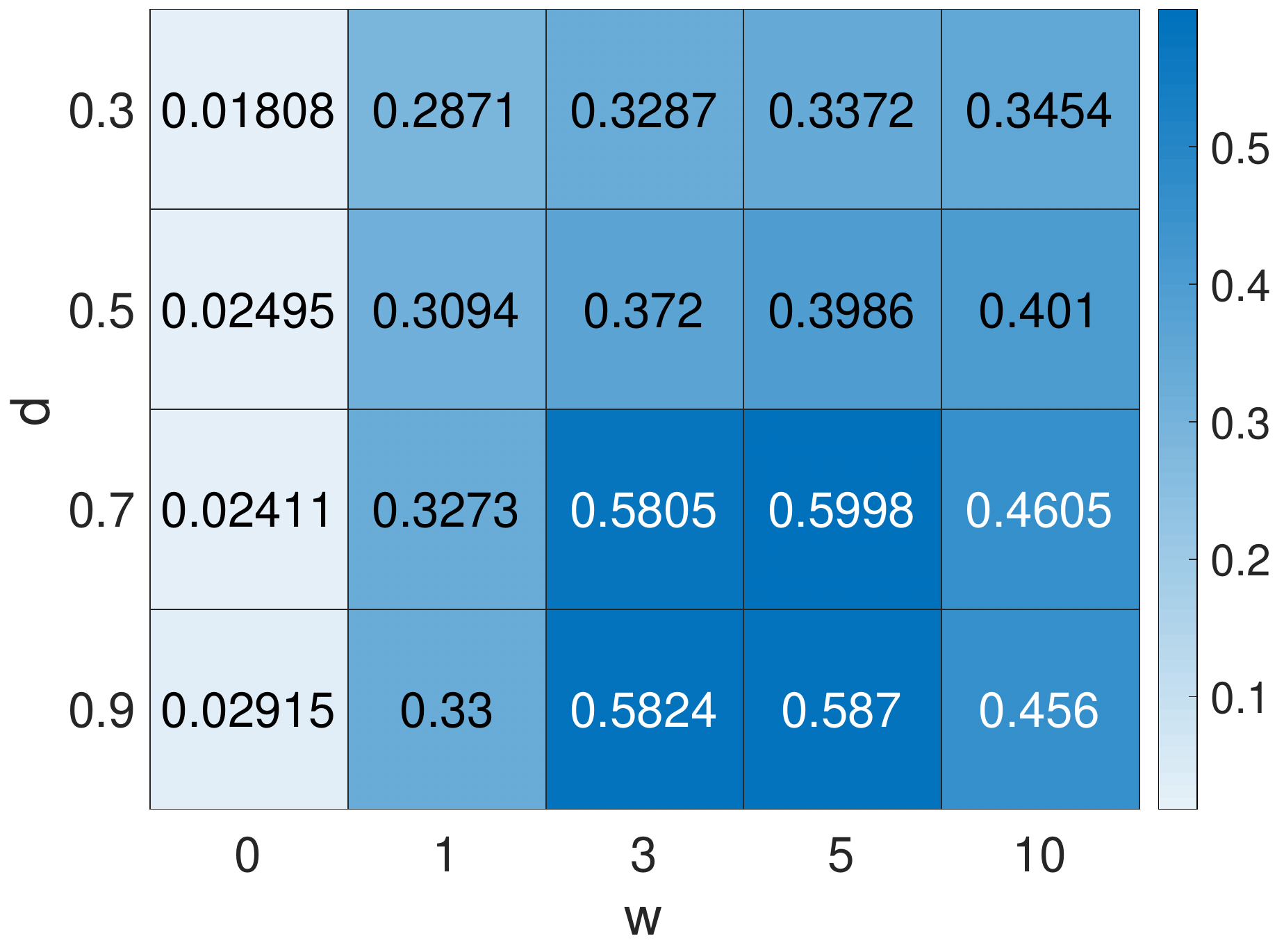}
  \label{fig:pmi_introduction}
} \hspace{1mm}
\subfloat [Average PMI score on the \\ Lectures dataset.]{
  \includegraphics[width=.31\linewidth]{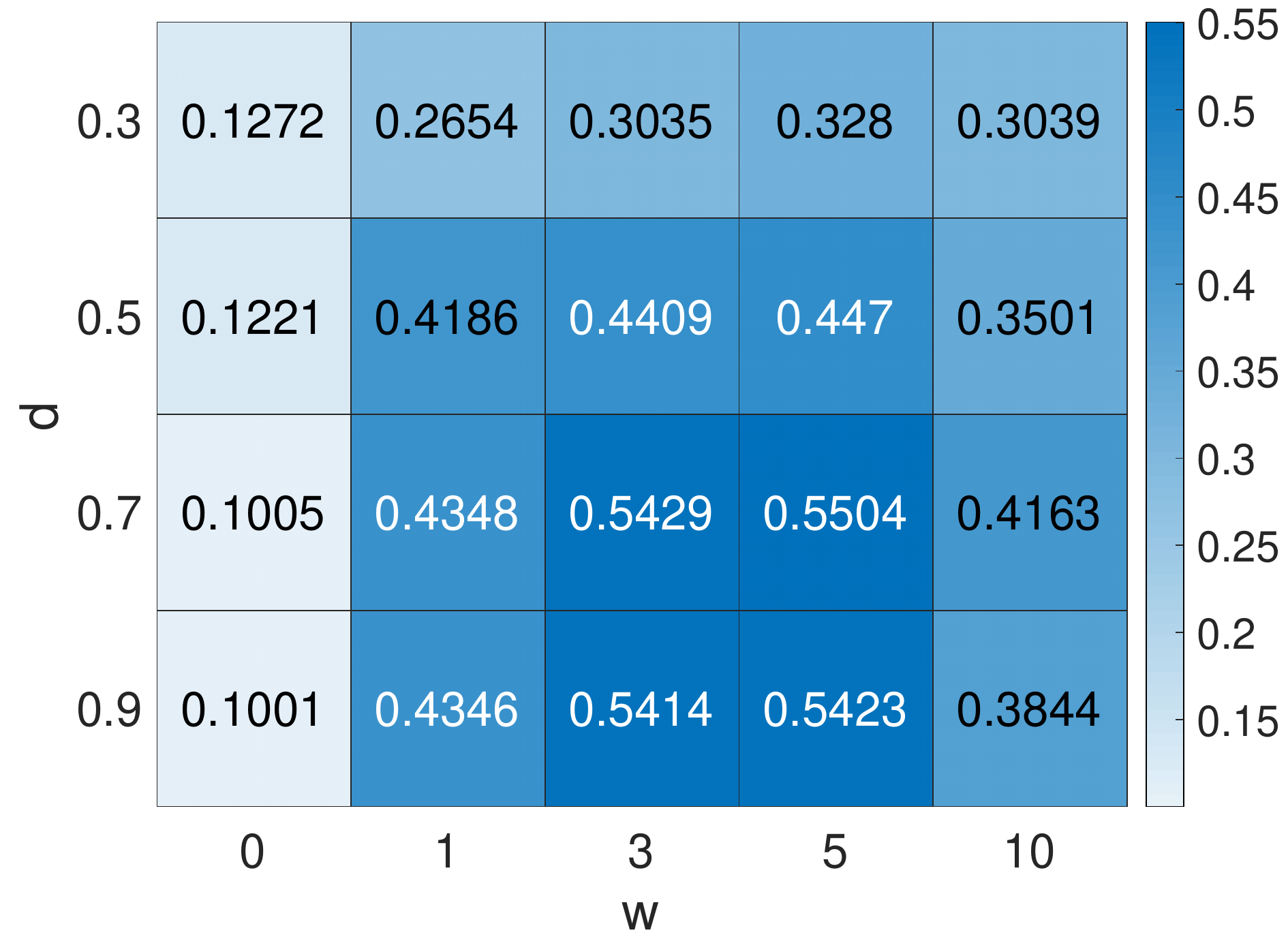}
  \label{fig:pmi_lectures}
} \hspace{1mm}
\subfloat [Average PMI score on the \\ Textbook dataset.]{
  \includegraphics[width=.31\linewidth]{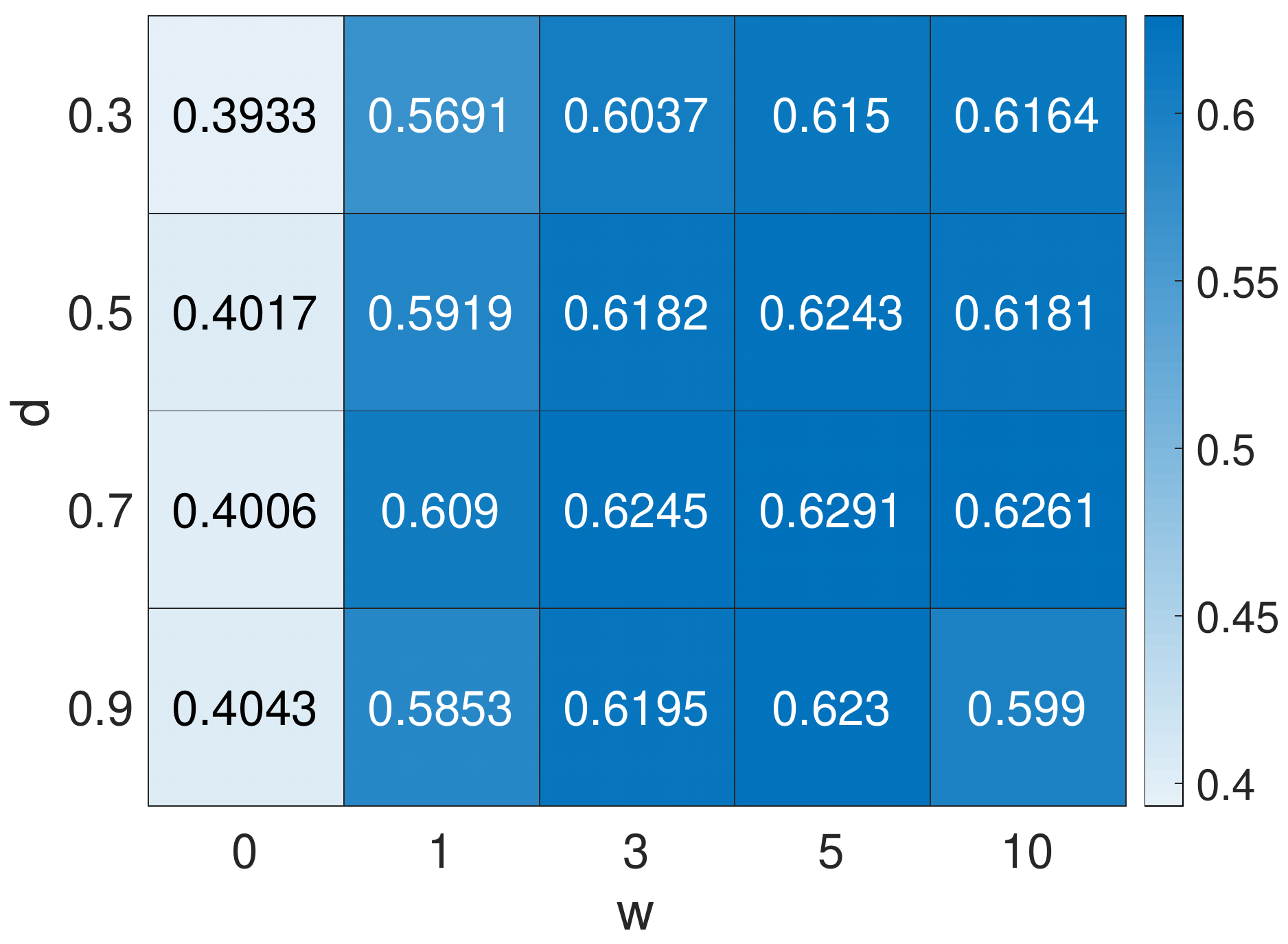}
  \label{fig:pmi_clinical}
}%
\vspace{-4mm}
\captionsetup{font=small}
  \caption{Heatmaps for average PMI scores across the the documents in Introduction, Lectures and Textbook dataset with varying $w$ and $d$. We see that $w = 5, d = 0.7$ leads to the highest PMI scores, consistent with our choices from visually inspecting Figure \ref{fig:wd_svd_heatmap}.}
  \label{fig:wd_PMI}
\end{figure}

We note that because BATS prevents overlapping words by design, it will achieve a similarity score of $0$ for both Jacc and Dice. For this reason, we will not use these metrics to claim that BATS has the ``best" performance on the topic modeling task (the difference between ``low overlap'' and ``no overlap'' is not likely important). Instead, we include Jacc and Dice for BATS and the baselines to provide a more holistic picture of the baselines. As some baselines will have a high Jacc or Dice score as well as strong performance on other metrics, these scores highlight a potential tradeoff between low overlap and high performance on the coherence metrics.
As a result, we also define composite metrics for evaluation which penalize the coherence scores on pairs of topics according to the similarity scores. Specifically, using ${\tt PMI}_i$ and ${\tt UMass}_i$ as the coherence scores for topic $i$ and ${\tt sim}_{i,j}$ as the similarity score between topics $i$ and $j$ according to Jacc or Dice, we compute:
\vspace{-1mm}
\begin{equation} \label{eq_pmisim}
    PMI^{sim} = \frac{\sum^k_{i=1}\sum^k_{j=1}({\tt PMI}_i + {\tt PMI}_j)(1-(\sgn({\tt PMI}_i + {\tt PMI}_j){\tt sim}_{i,j}))/2}{k^2},
\end{equation}
\begin{equation} \label{eq_umsim}
    UMass^{sim} = \frac{\sum^k_{i=1}\sum^k_{j=1}({\tt UMass}_i + {\tt UMass}_j)(1+{\tt sim}_{i,j})/2}{k^2},
\end{equation}
where $k$ refers to the total number of topics and $\sgn(x)$ is the sign function that evaluates to 1 if $x \geq 0$ and -1 if $x < 0$. \eqref{eq_pmisim} uses the $\sgn$ function to ensure that pairs of topics with high similarity scores are penalized regardless of the sign of the sum of their PMI scores. Because the UMass scores are always negative, we can simplify this penalty to $1+{\tt sim}_{i,j}$ in \eqref{eq_umsim}.

\vspace{-2mm}
\subsubsection{Hyperparameter analysis}
\label{sssec:modeling_hyperparameter}
In Section \ref{sssec:SVD_heatmap}, we determined values for $w$ and $d$ based on visually inspecting heatmaps of pairwise distances between the SVD sentence vectors. We now test these parameters against the PMI metric to validate our choices on the topic modeling task. Looking at Table \ref{fig:wd_PMI}, we can see that the best choices of $w$ and $d$ are (i) consistent with our earlier analysis (i.e., $d \approx 0.7$, $w = 3$ or $5$) and (ii) reasonably consistent across datasets. This validates the idea that BATS can be fully unsupervised, where we do not have an initial training procedure for hyperparameters.

Additionally, for completeness, we study the impact of changing $\mathcal{K}$ in the PMI score across the Textbook, Lectures, and Introduction datasets for BATS. The results are shown in Table \ref{tab:changing_K}. We find that increasing $\mathcal{K}$ decreases the PMI score: this makes intuitive sense because as the number of words in the topics increases, less relevant words should be included, if the words within each topic are ranked effectively. In BATS, our ranking is done according to the TF-IDF scores.

\vspace{-2mm}
\subsubsection{Results and discussion}
\label{sssec:modeling_discussion}
The results obtained by each algorithm on the five datasets are given in Table \ref{tab14}. We present the mean and standard deviations on topic diversity, topic coherence, and the four cases of joint metrics. The first two columns, Jacc and Dice, indicate the diversities of the topics (smaller being better). The following columns then give the topic coherence scores, PMI and UMass (larger being better), followed by their combinations with the similarity measures (e.g., $PMI^{Dice}$ is PMI with Dice used for ${\tt sim}_{ij}$ in \eqref{eq_pmisim}).

\bgroup
\def\arraystretch{1}
\begin{table*}[t]
\small
\begin{center}
  \captionsetup{font=small}
  \caption{\small{Performance of each algorithm on the Textbook, Lectures, Introductions, \qq{Wiki, and News} datasets in terms of topic coherence, similarity, and composite metrics. Our algorithm has the highest performance on most of the metrics, indicating it achieves the best balance between topic coherence and diversity.}}
\vspace{-3mm}
\resizebox{1.0\linewidth}{!}
{
\begin{tabular}{ccccccccc}
\hline
\multicolumn{9}{c}{Textbook Dataset} \\ \hline
\multicolumn{1}{c|}{} &
  \multicolumn{1}{c|}{Jacc} &
  \multicolumn{1}{c|}{Dice} &
  \multicolumn{1}{c|}{PMI} &
  \multicolumn{1}{c|}{$PMI^{Jacc}$} &
  \multicolumn{1}{c|}{$PMI^{Dice}$} &
  \multicolumn{1}{c|}{UMass} &
  \multicolumn{1}{c|}{$UMass^{Jacc}$} &
  $UMass^{Dice}$ \\ \hline
\multicolumn{1}{c|}{LDA} &
  \multicolumn{1}{c|}{$\mathbf{0.00 \pm 0.00}$} &
  \multicolumn{1}{c|}{$\mathbf{0.00 \pm 0.00}$} &
  \multicolumn{1}{c|}{$-0.71 \pm 0.91$} &
  \multicolumn{1}{c|}{$-0.71 \pm 0.91$} &
  \multicolumn{1}{c|}{$-0.71 \pm 0.91$} &
  \multicolumn{1}{c|}{$-14.74 \pm 2.78$} &
  \multicolumn{1}{c|}{$-14.74 \pm 2.78$} &
  $-14.74 \pm 2.78$ \\ \hline
\multicolumn{1}{c|}{HDP} &
  \multicolumn{1}{c|}{$0.01 \pm 0.01$} &
  \multicolumn{1}{c|}{$0.02 \pm 0.02$} &
  \multicolumn{1}{c|}{$-6.18 \pm 0.86$} &
  \multicolumn{1}{c|}{$-6.24 \pm 1.04$} &
  \multicolumn{1}{c|}{$-6.29 \pm 1.05$} &
  \multicolumn{1}{c|}{$-22.30 \pm 0.61$} &
  \multicolumn{1}{c|}{$-22.52 \pm 0.68$} &
  $-22.72 \pm 0.80$ \\ \hline
\multicolumn{1}{c|}{LSA} &
  \multicolumn{1}{c|}{$0.28 \pm 0.10$} &
  \multicolumn{1}{c|}{$0.42 \pm 0.13$} &
  \multicolumn{1}{c|}{$\mathbf{0.65 \pm 0.89}$} &
  \multicolumn{1}{c|}{$0.36 \pm 0.73$} &
  \multicolumn{1}{c|}{$0.22 \pm 0.72$} &
  \multicolumn{1}{c|}{$-8.11 \pm 2.21$} &
  \multicolumn{1}{c|}{$-10.38 \pm 2.96$} &
  $-11.49 \pm 3.30$ \\ \hline
\multicolumn{1}{c|}{pLSA} &
  \multicolumn{1}{c|}{$0.10 \pm 0.09$} &
  \multicolumn{1}{c|}{$0.14 \pm 0.12$} &
  \multicolumn{1}{c|}{$-3.63 \pm 1.52$} &
  \multicolumn{1}{c|}{$-3.97 \pm 1.63$} &
  \multicolumn{1}{c|}{$-4.14  \pm 1.72$} &
  \multicolumn{1}{c|}{$-14.65 \pm 2.82$} &
  \multicolumn{1}{c|}{$-16.06 \pm 3.34$} &
  $-16.78 \pm 3.77$ \\ \hline
\multicolumn{1}{c|}{NMF} &
  \multicolumn{1}{c|}{$0.21 \pm 0.10$} &
  \multicolumn{1}{c|}{$0.31 \pm 0.13$} &
  \multicolumn{1}{c|}{$-1.69\pm 1.41$} &
  \multicolumn{1}{c|}{$-1.45\pm 1.19$} &
  \multicolumn{1}{c|}{$-1.32 \pm 1.11$} &
  \multicolumn{1}{c|}{$-13.41 \pm 2.82$} &
  \multicolumn{1}{c|}{$-15.94 \pm 3.99$} &
  $-17.22 \pm 4.28$ \\ \hline
 \multicolumn{1}{c|}{SeaNMF}&
  \multicolumn{1}{c|}{$0.02 \pm 0.05$} &
  \multicolumn{1}{c|}{$0.05 \pm 0.08$} &
  \multicolumn{1}{c|}{$0.26 \pm 0.65$} &
  \multicolumn{1}{c|}{$0.24 \pm 0.64$} &
  \multicolumn{1}{c|}{$0.23 \pm 0.63$} &
  \multicolumn{1}{c|}{$\mathbf{-6.64 \pm 1.81}$} &
  \multicolumn{1}{c|}{$\mathbf{-6.85 \pm 2.06}$} &
  \multicolumn{1}{c}{$\mathbf{-7.01 \pm 2.24}$} \\ \hline
\multicolumn{1}{c|}{BATS} &
  \multicolumn{1}{c|}{$\mathbf{0.00 \pm 0.00}$} &
  \multicolumn{1}{c|}{$\mathbf{0.00 \pm 0.00}$} &
  \multicolumn{1}{c|}{$0.62 \pm 0.63$} &
  \multicolumn{1}{c|}{$\mathbf{0.62 \pm 0.63}$} &
  \multicolumn{1}{c|}{$\mathbf{0.62 \pm 0.63}$} &
  \multicolumn{1}{c|}{$-10.28 \pm 2.31$} &
  \multicolumn{1}{c|}{$-10.28 \pm 2.31$} &
  $-10.28 \pm 2.31$ \\ \hline
\multicolumn{9}{c}{Lectures Dataset} \\ \hline
\multicolumn{1}{c|}{} &
  \multicolumn{1}{c|}{Jacc} &
  \multicolumn{1}{c|}{Dice} &
  \multicolumn{1}{c|}{PMI} &
  \multicolumn{1}{c|}{$PMI^{Jacc}$} &
  \multicolumn{1}{c|}{$PMI^{Dice}$} &
  \multicolumn{1}{c|}{UMass} &
  \multicolumn{1}{c|}{$UMass^{Jacc}$} &
  $UMass^{Dice}$ \\ \hline
\multicolumn{1}{c|}{LDA} &
  \multicolumn{1}{c|}{$\mathbf{0.00 \pm 0.00}$} &
  \multicolumn{1}{c|}{$\mathbf{0.00 \pm 0.00}$} &
  \multicolumn{1}{c|}{$-0.71 \pm 0.85$} &
  \multicolumn{1}{c|}{$-0.71 \pm 0.85$} &
  \multicolumn{1}{c|}{$-0.71 \pm 0.85$} &
  \multicolumn{1}{c|}{$-14.60 \pm 2.52$} &
  \multicolumn{1}{c|}{$-14.60 \pm 2.52$} &
  $-14.60 \pm 2.52$ \\ \hline
\multicolumn{1}{c|}{HDP} &
  \multicolumn{1}{c|}{$0.01 \pm 0.01$} &
  \multicolumn{1}{c|}{$0.01 \pm 0.01$} &
  \multicolumn{1}{c|}{$-6.89 \pm 0.69$} &
  \multicolumn{1}{c|}{$-6.28 \pm 0.76$} &
  \multicolumn{1}{c|}{$-6.32 \pm 0.77$} &
  \multicolumn{1}{c|}{$-21.55 \pm 0.34$} &
  \multicolumn{1}{c|}{$-21.72 \pm 0.37$} &
  $-21.86 \pm 0.42$ \\ \hline
\multicolumn{1}{c|}{LSA} &
  \multicolumn{1}{c|}{$0.27 \pm 0.07$} &
  \multicolumn{1}{c|}{$0.41 \pm 0.09$} &
  \multicolumn{1}{c|}{$\mathbf{0.59 \pm 0.62}$} &
  \multicolumn{1}{c|}{$0.35 \pm 0.53$} &
  \multicolumn{1}{c|}{$0.24 \pm 0.49$} &
  \multicolumn{1}{c|}{$\mathbf{-7.13 \pm 1.76}$} &
  \multicolumn{1}{c|}{$-8.97 \pm 2.12$} &
  $-9.92 \pm 2.34$ \\ \hline
\multicolumn{1}{c|}{pLSA} &
  \multicolumn{1}{c|}{$0.04 \pm 0.03$} &
  \multicolumn{1}{c|}{$0.07 \pm 0.04$} &
  \multicolumn{1}{c|}{$-6.4 \pm 0.77$} &
  \multicolumn{1}{c|}{$-6.7\pm 0.86$} &
  \multicolumn{1}{c|}{$-6.9 \pm 0.94$} &
  \multicolumn{1}{c|}{$-19.44 \pm 0.78$} &
  \multicolumn{1}{c|}{$-20.27 \pm 1.04$} &
  $-20.90 \pm 1.31$ \\ \hline
\multicolumn{1}{c|}{NMF} &
  \multicolumn{1}{c|}{$0.26 \pm 0.08$} &
  \multicolumn{1}{c|}{$0.39 \pm 0.10$} &
  \multicolumn{1}{c|}{$0.48 \pm 0.52$} &
  \multicolumn{1}{c|}{$0.30 \pm 0.46$} &
  \multicolumn{1}{c|}{$0.21 \pm 0.44$} &
  \multicolumn{1}{c|}{$-9.07 \pm 2.30$} &
  \multicolumn{1}{c|}{$-11.39 \pm 2.84$} &
  $-12.51 \pm 3.10$ \\ \hline
  \multicolumn{1}{c|}{SeaNMF}&
  \multicolumn{1}{c|}{$0.02 \pm 0.03$} &
  \multicolumn{1}{c|}{$0.04 \pm 0.06$} &
  \multicolumn{1}{c|}{$0.26 \pm 0.41$} &
  \multicolumn{1}{c|}{$0.21 \pm 0.40$} &
  \multicolumn{1}{c|}{$0.19 \pm 0.40$} &
  \multicolumn{1}{c|}{$-8.36 \pm 1.16$} &
  \multicolumn{1}{c|}{$\mathbf{-8.49 \pm 1.18}$} &
   $\mathbf{-8.59 \pm 1.20}$ \\ \hline
\multicolumn{1}{c|}{BATS} &
  \multicolumn{1}{c|}{$\mathbf{0.00 \pm 0.00}$} &
  \multicolumn{1}{c|}{$\mathbf{0.00 \pm 0.00}$} &
  \multicolumn{1}{c|}{$0.54 \pm 0.53$}&
  \multicolumn{1}{c|}{$\mathbf{0.54 \pm 0.53}$} &
  \multicolumn{1}{c|}{$\mathbf{0.54 \pm 0.53}$} &
  \multicolumn{1}{c|}{$-9.08 \pm 1.82$} &
  \multicolumn{1}{c|}{$-9.08 \pm 1.82$} &
  $-9.08 \pm 1.82$ \\ \hline
\multicolumn{9}{c}{Introductions Dataset} \\ \hline
\multicolumn{1}{c|}{} &
  \multicolumn{1}{c|}{Jacc} &
  \multicolumn{1}{c|}{Dice} &
  \multicolumn{1}{c|}{PMI} &
  \multicolumn{1}{c|}{$PMI^{Jacc}$} &
  \multicolumn{1}{c|}{$PMI^{Dice}$} &
  \multicolumn{1}{c|}{UMass} &
  \multicolumn{1}{c|}{$UMass^{Jacc}$} &
  $UMass^{Dice}$ \\ \hline
\multicolumn{1}{c|}{LDA} &
  \multicolumn{1}{c|}{$\mathbf{0.00 \pm 0.00}$} &
  \multicolumn{1}{c|}{$\mathbf{0.00 \pm 0.00}$} &
  \multicolumn{1}{c|}{$-2.09 \pm 0.84$} &
  \multicolumn{1}{c|}{$-2.09 \pm 0.84$} &
  \multicolumn{1}{c|}{$-2.09 \pm 0.84$} &
  \multicolumn{1}{c|}{$-15.38 \pm 1.72$} &
  \multicolumn{1}{c|}{$-15.38 \pm 1.72$} &
  $-15.38 \pm 1.72$ \\ \hline
\multicolumn{1}{c|}{HDP} &
  \multicolumn{1}{c|}{$0.01 \pm 0.01$} &
  \multicolumn{1}{c|}{$0.01 \pm 0.02$} &
  \multicolumn{1}{c|}{$-7.13 \pm 1.14$} &
  \multicolumn{1}{c|}{$-7.17 \pm 1.44$} &
  \multicolumn{1}{c|}{$-7.21 \pm 1.14$} &
  \multicolumn{1}{c|}{$-21.78 \pm 1.61$} &
  \multicolumn{1}{c|}{$-21.92 \pm 1.61$} &
  $-22.04 \pm 1.62$ \\ \hline
\multicolumn{1}{c|}{LSA} &
  \multicolumn{1}{c|}{$0.21 \pm 0.09$} &
  \multicolumn{1}{c|}{$0.31 \pm 0.12$} &
  \multicolumn{1}{c|}{$-0.64 \pm 0.73$} &
  \multicolumn{1}{c|}{$-0.84 \pm 0.88$} &
  \multicolumn{1}{c|}{$-0.93 \pm 0.94$} &
  \multicolumn{1}{c|}{$-8.11 \pm 2.03$} &
  \multicolumn{1}{c|}{$-9.88 \pm 2.77$} &
  $-10.77 \pm 3.11$ \\ \hline
\multicolumn{1}{c|}{pLSA} &
  \multicolumn{1}{c|}{$0.04 \pm 0.05$} &
  \multicolumn{1}{c|}{$0.06 \pm 0.07$} &
  \multicolumn{1}{c|}{$-5.35 \pm 1.43$} &
  \multicolumn{1}{c|}{$-5.54 \pm 1.46$} &
  \multicolumn{1}{c|}{$-5.67 \pm 1.49$} &
  \multicolumn{1}{c|}{$-16.15 \pm 2.29$} &
  \multicolumn{1}{c|}{$16.76 \pm 2.39$} &
  $-17.16 \pm 2.56$ \\ \hline
\multicolumn{1}{c|}{NMF} &
  \multicolumn{1}{c|}{$0.21 \pm 0.08$} &
  \multicolumn{1}{c|}{$0.32 \pm 0.11$} &
  \multicolumn{1}{c|}{$-2.16\pm 1.15$} &
  \multicolumn{1}{c|}{$-2.56\pm 1.34$} &
  \multicolumn{1}{c|}{$-2.76 \pm 1.43$} &
  \multicolumn{1}{c|}{$-14.18 \pm 2.32$} &
  \multicolumn{1}{c|}{$-17.05 \pm 2.61$} &
  $-18.50 \pm 2.85$ \\ \hline
  \multicolumn{1}{c|}{SeaNMF}&
  \multicolumn{1}{c|}{$0.01 \pm 0.02 $} &
  \multicolumn{1}{c|}{$0.02 \pm 0.05 $} &
  \multicolumn{1}{c|}{$0.38 \pm 0.35 $} &
  \multicolumn{1}{c|}{$0.36 \pm 0.36 $} &
  \multicolumn{1}{c|}{$0.35 \pm 0.42$} &
  \multicolumn{1}{c|}{$-8.18 \pm 0.85$} &
  \multicolumn{1}{c|}{$-8.23 \pm 0.86$} &
   $-8.27 \pm 0.87$ \\ \hline
\multicolumn{1}{c|}{BATS} &
  \multicolumn{1}{c|}{$\mathbf{0.00 \pm 0.00}$} &
  \multicolumn{1}{c|}{$\mathbf{0.00 \pm 0.00}$} &
  \multicolumn{1}{c|}{$\mathbf{0.58 \pm 0.45}$} &
  \multicolumn{1}{c|}{$\mathbf{0.58 \pm 0.45}$} &
  \multicolumn{1}{c|}{$\mathbf{0.58 \pm 0.45}$} &
  \multicolumn{1}{c|}{$\mathbf{-7.69 \pm 1.63}$} &
  \multicolumn{1}{c|}{$\mathbf{-7.69 \pm 1.63}$} &
  $\mathbf{-7.69 \pm 1.63}$ \\ \hline
\multicolumn{9}{c}{\qq{Wiki Dataset}} \\ \hline
\multicolumn{1}{c|}{} &
  \multicolumn{1}{c|}{\qq{Jacc}} &
  \multicolumn{1}{c|}{\qq{Dice}} &
  \multicolumn{1}{c|}{\qq{PMI}} &
  \multicolumn{1}{c|}{\qq{$PMI^{Jacc}$}} &
  \multicolumn{1}{c|}{\qq{$PMI^{Dice}$}} &
  \multicolumn{1}{c|}{\qq{UMass}} &
  \multicolumn{1}{c|}{\qq{$UMass^{Jacc}$}} &
  \qq{$UMass^{Dice}$} \\ \hline
\multicolumn{1}{c|}{\qq{LDA}} &
  \multicolumn{1}{c|}{$\mathbf{0.00 \pm 0.00}$} &
  \multicolumn{1}{c|}{$\mathbf{0.00 \pm 0.00}$} &
  \multicolumn{1}{c|}{$-3.55 \pm 1.12$} &
  \multicolumn{1}{c|}{$-3.55 \pm 1.12$} &
  \multicolumn{1}{c|}{$-3.55\pm 1.12$} &
  \multicolumn{1}{c|}{$-19.14 \pm 2.20$} &
  \multicolumn{1}{c|}{$-19.14 \pm 2.20$} &
  $-19.14 \pm 2.20$ \\ \hline
\multicolumn{1}{c|}{\qq{HDP}} &
  \multicolumn{1}{c|}{$0.027 \pm 0.02$} &
  \multicolumn{1}{c|}{$0.05 \pm 0.04$} &
  \multicolumn{1}{c|}{$-5.98 \pm 1.06$} &
  \multicolumn{1}{c|}{$-6.13 \pm 1.03$} &
  \multicolumn{1}{c|}{$-6.25 \pm 1.01$} &
  \multicolumn{1}{c|}{$-22.18 \pm 1.04$} &
  \multicolumn{1}{c|}{$-22.77 \pm 0.99$} &
  $-23.28 \pm 1.12$ \\ \hline
\multicolumn{1}{c|}{\qq{LSA}} &
  \multicolumn{1}{c|}{$0.19 \pm 0.10$} &
  \multicolumn{1}{c|}{$0.29 \pm 0.13$} &
  \multicolumn{1}{c|}{$-0.05 \pm 1.28$} &
  \multicolumn{1}{c|}{$-0.30 \pm 1.21$} &
  \multicolumn{1}{c|}{$-0.43 \pm 1.21$} &
  \multicolumn{1}{c|}{$-8.11 \pm 2.03$} &
  \multicolumn{1}{c|}{$-9.88 \pm 2.77$} &
  $-10.77 \pm 3.11$ \\ \hline
\multicolumn{1}{c|}{\qq{pLSA}} &
  \multicolumn{1}{c|}{$0.12 \pm 0.09$} &
  \multicolumn{1}{c|}{$0.17 \pm 0.12$} &
  \multicolumn{1}{c|}{$-1.97 \pm 1.63$} &
  \multicolumn{1}{c|}{$-2.30 \pm 1.77$} &
  \multicolumn{1}{c|}{$-2.46 \pm 1.87$} &
  \multicolumn{1}{c|}{$-13.44 \pm 3.57$} &
  \multicolumn{1}{c|}{$-15.05 \pm 4.29$} &
  $-15.85 \pm 4.73$ \\ \hline
\multicolumn{1}{c|}{\qq{NMF}} &
  \multicolumn{1}{c|}{$0.14 \pm 0.09$} &
  \multicolumn{1}{c|}{$0.21 \pm 0.14$} &
  \multicolumn{1}{c|}{$-2.60\pm 1.64$} &
  \multicolumn{1}{c|}{$-3.04\pm 1.08$} &
  \multicolumn{1}{c|}{$-3.28 \pm 1.91$} &
  \multicolumn{1}{c|}{$-16.81 \pm 2.85$} &
  \multicolumn{1}{c|}{$-19.63 \pm 3.15$} &
  $-21.22 \pm 3.49$ \\ \hline
  \multicolumn{1}{c|}{\qq{SeaNMF}}&
  \multicolumn{1}{c|}{$0.01 \pm 0.03 $} &
  \multicolumn{1}{c|}{$0.01 \pm 0.06 $} &
  \multicolumn{1}{c|}{$0.37 \pm 0.28 $} &
  \multicolumn{1}{c|}{$0.36 \pm 0.26 $} &
  \multicolumn{1}{c|}{$0.35 \pm 0.22$} &
  \multicolumn{1}{c|}{$-7.02 \pm 0.85$} &
  \multicolumn{1}{c|}{$-7.13 \pm 0.92$} &
   $-8.27 \pm 0.92$ \\ \hline
\multicolumn{1}{c|}{\qq{BATS}} &
  \multicolumn{1}{c|}{$\mathbf{0.00 \pm 0.00}$} &
  \multicolumn{1}{c|}{$\mathbf{0.00 \pm 0.00}$} &
  \multicolumn{1}{c|}{$\mathbf{0.44 \pm 0.24}$} &
  \multicolumn{1}{c|}{$\mathbf{0.44 \pm 0.24}$} &
  \multicolumn{1}{c|}{$\mathbf{0.44 \pm 0.24}$} &
  \multicolumn{1}{c|}{$\mathbf{-6.08 \pm  0.63}$} &
  \multicolumn{1}{c|}{$\mathbf{-6.08 \pm 0.63}$} &
  $\mathbf{-6.08 \pm 0.63}$ \\ \hline
  \multicolumn{9}{c}{\qq{News Dataset}} \\ \hline
\multicolumn{1}{c|}{} &
  \multicolumn{1}{c|}{\qq{Jacc}} &
  \multicolumn{1}{c|}{\qq{Dice}} &
  \multicolumn{1}{c|}{\qq{PMI}} &
  \multicolumn{1}{c|}{\qq{$PMI^{Jacc}$}} &
  \multicolumn{1}{c|}{\qq{$PMI^{Dice}$}} &
  \multicolumn{1}{c|}{\qq{UMass}} &
  \multicolumn{1}{c|}{\qq{$UMass^{Jacc}$}} &
  \qq{$UMass^{Dice}$} \\ \hline
\multicolumn{1}{c|}{\qq{LDA}} &
  \multicolumn{1}{c|}{$\mathbf{0.00 \pm 0.00}$} &
  \multicolumn{1}{c|}{$\mathbf{0.00 \pm 0.00}$} &
  \multicolumn{1}{c|}{$-2.88 \pm 1.32$} &
  \multicolumn{1}{c|}{$-2.88 \pm 1.32$} &
  \multicolumn{1}{c|}{$-2.88 \pm 1.32$} &
  \multicolumn{1}{c|}{$-18.77 \pm 2.14$} &
  \multicolumn{1}{c|}{$-18.77 \pm 2.14$} &
  $-18.77 \pm 2.14$ \\ \hline
\multicolumn{1}{c|}{\qq{HDP}} &
  \multicolumn{1}{c|}{$0.03 \pm 0.03$} &
  \multicolumn{1}{c|}{$0.07 \pm 0.06$} &
  \multicolumn{1}{c|}{$-5.25 \pm 1.48$} &
  \multicolumn{1}{c|}{$-5.42 \pm 1.46$} &
  \multicolumn{1}{c|}{$-5..56 \pm 1.46$} &
  \multicolumn{1}{c|}{$-21.71 \pm 1.36$} &
  \multicolumn{1}{c|}{$-22.52 \pm 1.25$} &
  $-23.19 \pm 1.42$ \\ \hline
\multicolumn{1}{c|}{\qq{LSA}} &
  \multicolumn{1}{c|}{$0.15 \pm 0.09$} &
  \multicolumn{1}{c|}{$0.24 \pm 0.12$} &
  \multicolumn{1}{c|}{$0.68 \pm 1.46$} &
  \multicolumn{1}{c|}{$0.46 \pm 1.45$} &
  \multicolumn{1}{c|}{$0.32 \pm 1.45$} &
  \multicolumn{1}{c|}{$-8.76 \pm 3.77$} &
  \multicolumn{1}{c|}{$-9.08 \pm 3.76$} &
  $-9.63 \pm 3.75 $ \\ \hline
\multicolumn{1}{c|}{\qq{pLSA}} &
  \multicolumn{1}{c|}{$0.11 \pm 0.09$} &
  \multicolumn{1}{c|}{$0.16 \pm 0.12$} &
  \multicolumn{1}{c|}{$-1.46 \pm 2.04 $} &
  \multicolumn{1}{c|}{$-1.78 \pm 2.15$} &
  \multicolumn{1}{c|}{$-1.92  \pm 2.24$} &
  \multicolumn{1}{c|}{$-12.37 \pm 4.05 $} &
  \multicolumn{1}{c|}{$-13.82 \pm 4.86$} &
  $-14.54 \pm 5.29$ \\ \hline
\multicolumn{1}{c|}{\qq{NMF}} &
  \multicolumn{1}{c|}{$ 0.17 \pm 0.08$} &
  \multicolumn{1}{c|}{$0.27 \pm 0.11$} &
  \multicolumn{1}{c|}{$-2.03 \pm 1.52$} &
  \multicolumn{1}{c|}{$-2.40 \pm 1.63$} &
  \multicolumn{1}{c|}{$-2.60 \pm 1.70 $} &
  \multicolumn{1}{c|}{$-17.21 \pm 2.84$} &
  \multicolumn{1}{c|}{$-20.14 \pm 3.20$} &
  $-21.78 \pm 3.51$ \\ \hline
  \multicolumn{1}{c|}{\qq{SeaNMF}}&
  \multicolumn{1}{c|}{$0.03 \pm 0.06 $} &
  \multicolumn{1}{c|}{$0.04 \pm 0.09 $} &
  \multicolumn{1}{c|}{$0.69 \pm 1.81 $} &
  \multicolumn{1}{c|}{$0.65 \pm 1.80 $} &
  \multicolumn{1}{c|}{$0.64 \pm 1.80 $} &
  \multicolumn{1}{c|}{$-8.92 \pm 2.92 $} &
  \multicolumn{1}{c|}{$-9.12 \pm 2.97$} &
   $-9.20 \pm 2.98$ \\ \hline
\multicolumn{1}{c|}{\qq{BATS}} &
  \multicolumn{1}{c|}{$\mathbf{0.00 \pm 0.00}$} &
  \multicolumn{1}{c|}{$\mathbf{0.00 \pm 0.00}$} &
  \multicolumn{1}{c|}{$\mathbf{0.76 \pm 1.49}$} &
  \multicolumn{1}{c|}{$\mathbf{0.76 \pm 1.49}$} &
  \multicolumn{1}{c|}{$\mathbf{0.76 \pm 1.49}$} &
  \multicolumn{1}{c|}{$\mathbf{-7.56 \pm  2.73}$} &
  \multicolumn{1}{c|}{$\mathbf{-7.56 \pm  2.73}$} &
  $\mathbf{-7.56 \pm  2.73}$ \\ \hline
\end{tabular}}
\label{tab14}
\end{center}
\end{table*}
\egroup

Overall, we see that compared with the baselines, \textit{our method BATS obtains competitive topic coherence scores, the lowest similarity scores, and the best composite scores in most cases}. For the Introductions, \qq{Wiki, and News} datasets, \qq{which are the three largest we consider,} our method maintains higher performance than all baselines in all metrics. On the Textbook and Lectures datasets, BATS achieves the best performance on the PMI composite metrics and is a close second to LSA on the non-composite PMI. BATS performs second only to SeaNMF for the UMass composite metrics and comparably to the best baselines on the non-composite UMass metric. The baseline which tends to outperform our algorithm in terms of topic coherence, LSA, also performs the worst in terms of topic diversity. To interpret this diversity performance, we note that a Jacc score of $0.25$ and a Dice score of $0.4$ correspond roughly to $|T_i \cap T_j| \propto 0.4$ in \eqref{eq_jc}, i.e., a 40\% duplication between topics. Thus, LSA (as well as NMF) usually has up to 40\% average overlap in topic words, leading to confusing topics, while our method yields no noticeable overlap. On the other hand, the baseline which matches our algorithm in topic diversity, LDA, is among the lowest performing in terms of coherence, which is also reflected in the composite metrics.

Although SeaNMF achieves both a high topic diversity score and UMass score, on the PMI metric it falls substantially short of BATS (for all datasets) and LSA (for the Textbook and Lectures datasets). This may be due to the fact that, for our purposes, the UMass score for each document is focused only on topics in that document whereas the PMI score includes information from the rest of the dataset. We can thus conclude that, among the algorithms tested, \textit{our algorithm finds the best balance between topic coherence and diversity} for single document topic modeling; its consistent performance across the datasets also shows that it is robust to variations in dataset properties.

We also observe an interesting pattern in the baselines: the spectral methods -- LSA and NMF -- perform high in coherence but low in similarity, while the generative models -- LDA and pLSA -- have the opposite trends. While spectral approaches can extract topics that are interpretable when taken individually, there is high similarity between them because they are based on matrix decomposition and do not consider diversity. Generative models can extract diversified topics, but when they are operating on single documents with few word co-occurrences, the resulting topics will not be as coherent. These observations are consistent with Figure \ref{fig_word}.

\vspace{-2mm}
\subsection{Text Segmentation}
\label{ssec:res-seg}
\vspace{-.6mm}

\subsubsection{Baselines}
We compared BATS against six baselines for the text segmentation task:

\noindent \textit{(i) TextTiling} \cite{hearst1997texttiling}: TextTiling divides the text into pseudosentences, assigns similarity scores at the gaps, detects peak differences in the scores, and marks the peaks as boundaries. The boundaries are normalized to the closest sentence breaks. We use the implementation from the {\tt nltk} package.

\noindent \textit{(ii) C99} \cite{choi2000advances}: C99 is another popular text segmentation algorithm that inserts boundaries based on average inter-sentence similarities. More specifically, a ranking transformation is performed, pairwise cosine distances between sentences are computed based on the ranking, and boundaries are determined based on these similarities. We implement C99 de-novo in Python.


\noindent \textit{(iii) Modified DP Algorithm with LDA (LDA\_MDP)} \cite{misra2011text}: LDA\_MDP performs text segmentation based the LDA topic model, with the segmentation being implemented with dynamic processing (DP) techniques. The method has also been tested using an alternate topic model, multinomial mixture, but LDA has has been found to obtain better performance. We implement the LDA\_MDP model de-novo in Python.

\noindent \textit{(iv) TopicTiling} \cite{riedl2012topictiling}: TopicTiling is based on TextTiling, with additionally integrated topic information from the LDA topic model for text segmentation. We implement TopicTiling de-novo in Python, using a window size of 2 and 500 iterations.

\noindent \textit{(v) SupervisedSeg}~\cite{koshorek2018text}: SupervisedSeg formulates text segmentation as a supervised learning task, training a hierarchical bidirectional neural LSTM model on the WIKI-727K dataset. The authors transform the text into word embeddings using the GoogleNews word2vec pre-trained model, and use the word embeddings as inputs to the neural model. SupervisedSeg then predicts a cut-off probability for each sentence. We used the open source pre-trained model available from~\cite{koshorek2018text}. 

\noindent \textit{(vi) BERTSeg}: To test the efficacy of pre-trained models, we designed BERTSeg, an algorithm that leverages the state-of-the-art Bidirectional Encoder Representations from Transformers (BERT) \cite{devlin2018bert} contextualized word embeddings for text segmentation. As the name suggests, BERT uses the Transformer deep learning architecture \cite{vaswani2017attention} to learn representations of words from unlabeled text, considering both the right context and left context in every layer (i.e., learning bidirectionally). BERTSeg is based on the open source Pytorch BERT model~\cite{reimers2019sentence}. After using the pre-trained BERT model to generate sentence embeddings, BERTSeg employs the agglomerative segment clustering method from BATS to convert the sentence embeddings into segments.

To say consistent across the algorithms, for the topic-based text segmentation methods -- LDA\_MDP and TopicTiling -- we train the topic model with the single document being segmented.

\begin{figure}[t!]
\centering
\subfloat [Average WD score on the \\ Introduction dataset.]{
  \includegraphics[width=.30\linewidth]{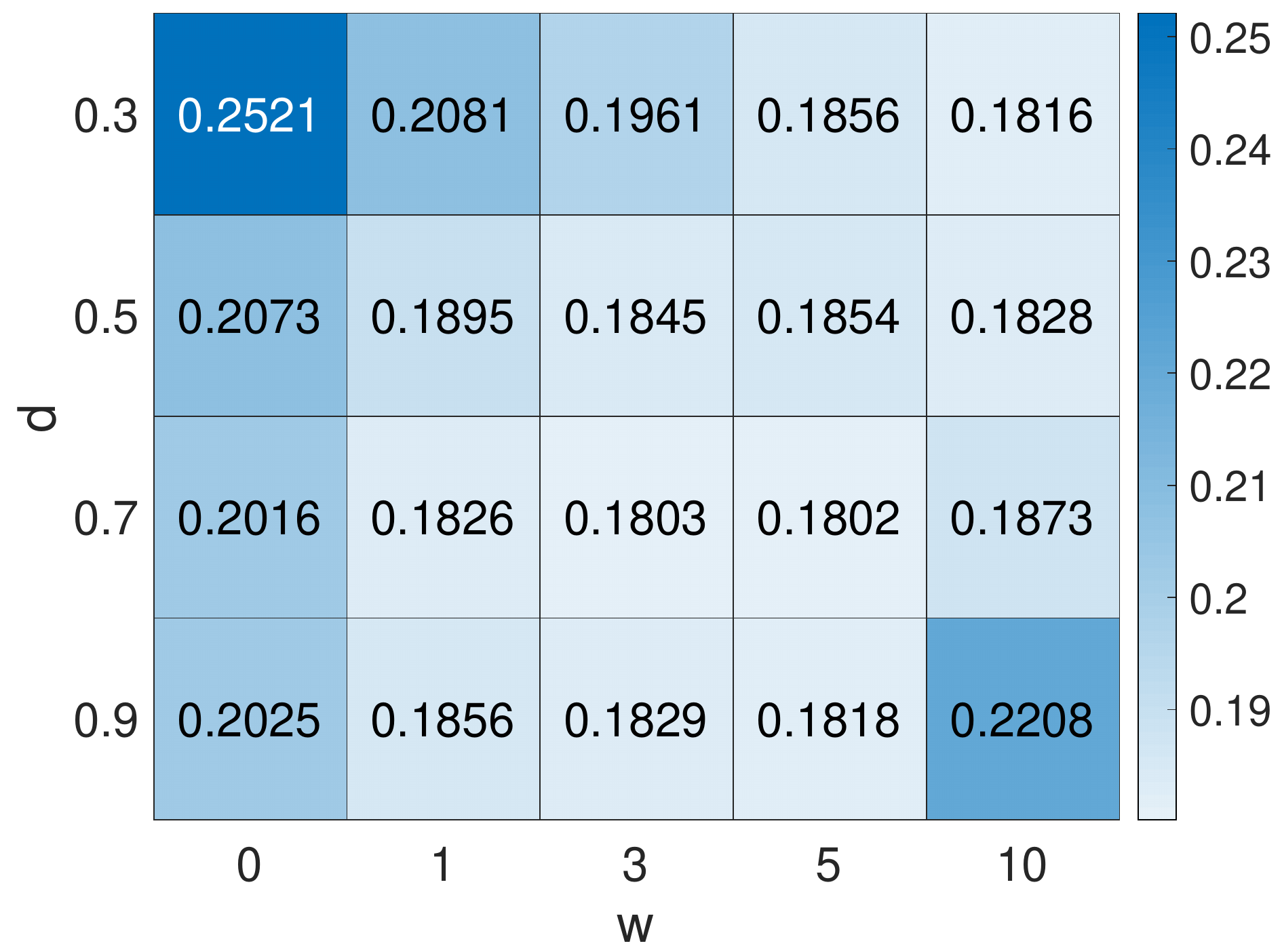}
  \label{fig:pmi_introduction}
}
\subfloat [Average WD score on the \\ Lectures dataset.]{
  \includegraphics[width=.31\linewidth]{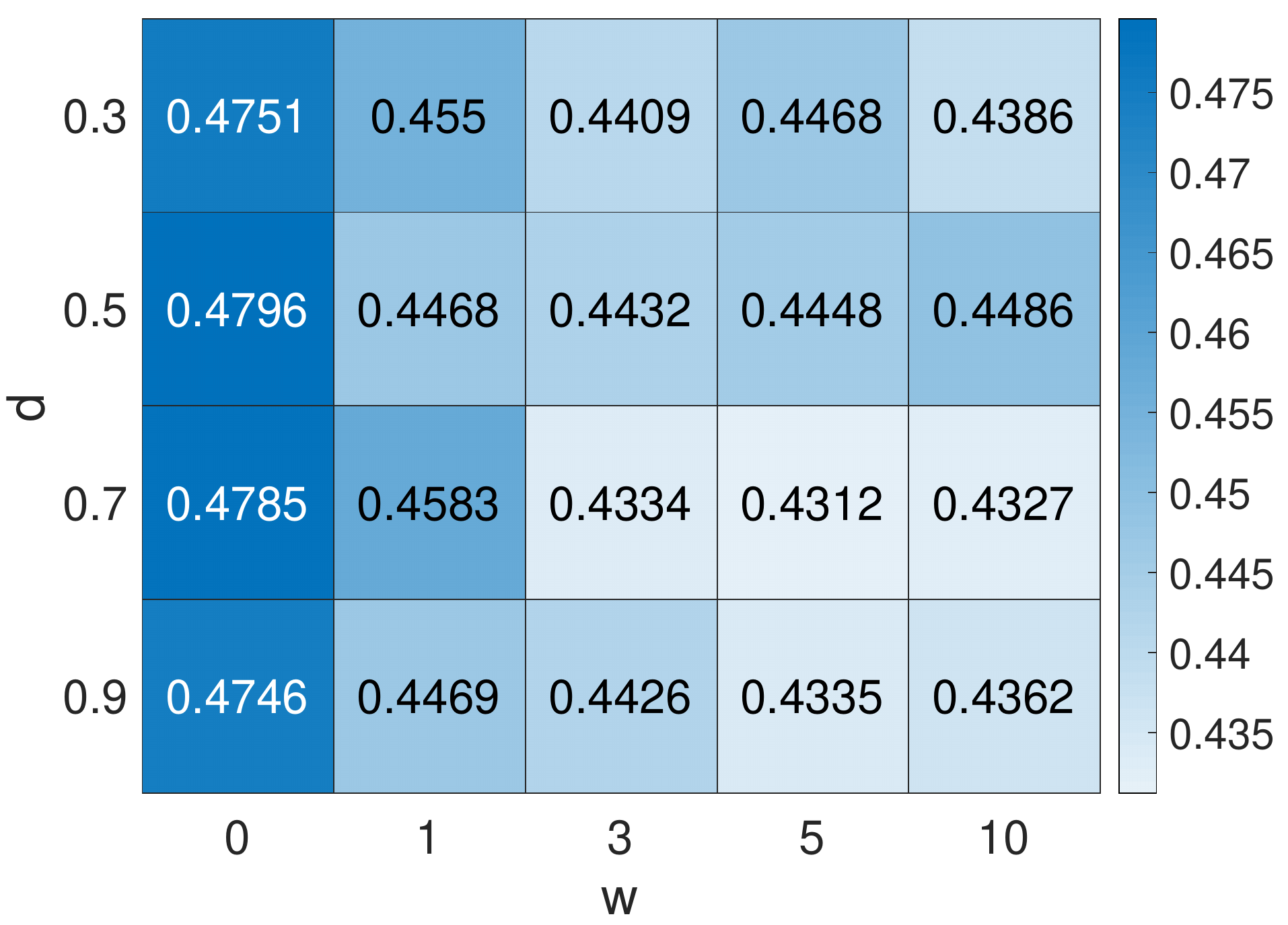}
  \label{fig:pmi_lectures}
}
\subfloat [Average WD score on the \\ Textbook dataset.]{
  \includegraphics[width=.31\linewidth]{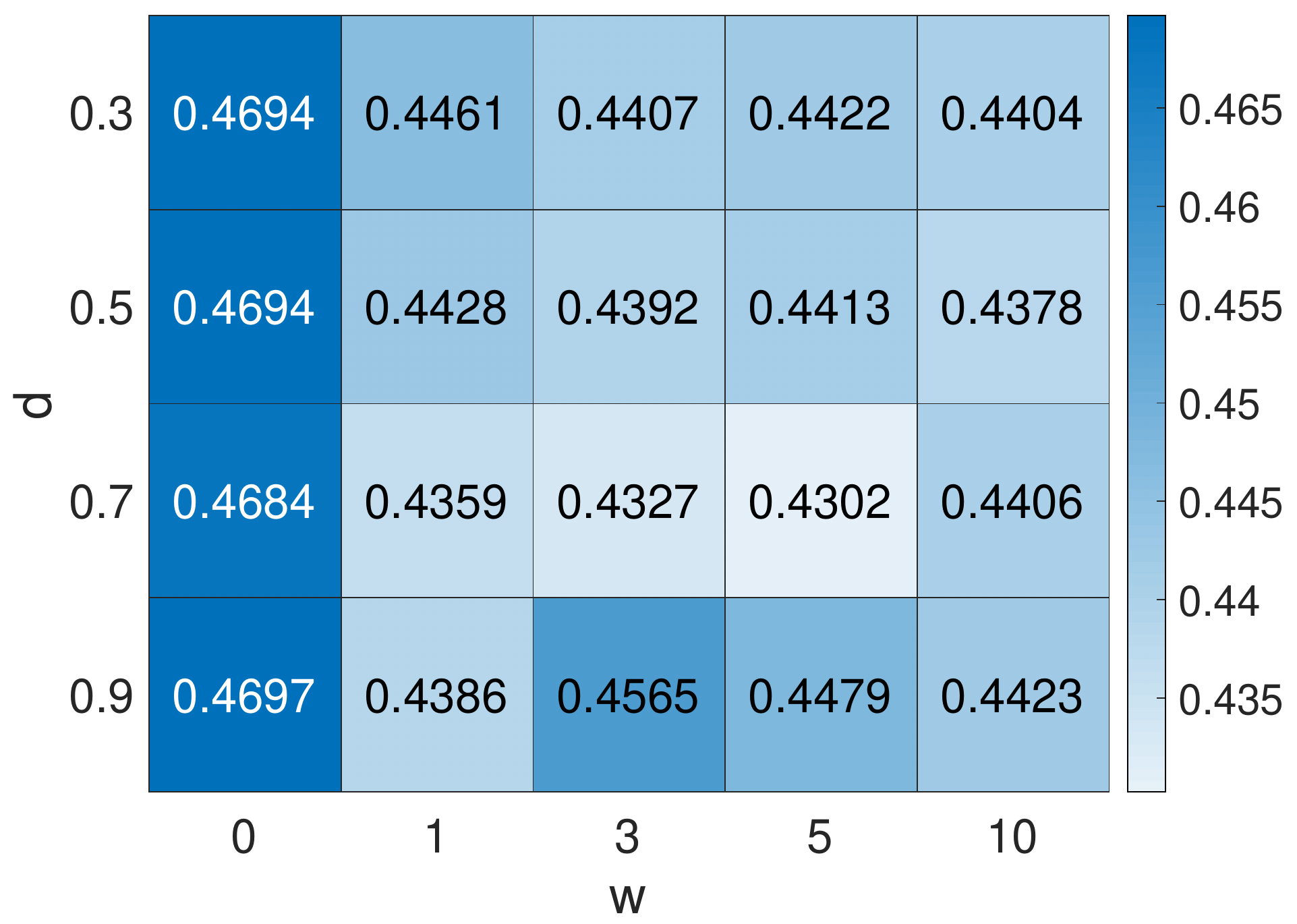}
  \label{fig:pmi_clinical}
}%
\vspace{-4mm}
  \caption{\small{Heatmaps for average WD scores across the the documents in Introduction, Lectures and Textbook dataset with varying $w$ and $d$. We find that $w = 5, d = 0.7$ leads to the best performance, consistent with the PMI results in Figure \ref{fig:wd_PMI}.}}
  \label{fig:wd_WD}
\end{figure}

\vspace{-2mm}
\subsubsection{Evaluation metrics}
We consider two standard text segmentation metrics, $P_k$ \cite{beeferman1999statistical} and WindowDiff (WD) \cite{pevzner2002critique}. Lower values indicate better performance. Each of these metrics compares the ground truth (i.e., reference) segmentation {\tt ref} to the estimated (i.e., hypothesized) segmentation {\tt hyp}. The $P_k$ metric calculates the number of disagreements in the positions of segment boundaries between {\tt hyp} and {\tt ref}; in doing so, it ignores the exact number of boundaries to be detected, and weights false positives more heavily \cite{pevzner2002critique}. WD, on the other hand, slides a fixed-sized window across the document, calculates the number of boundaries within that window, and records an error if {\tt ref} and {\tt hyp} disagree on the number.

Formally, let $(x_1, x_2, ..., x_N)$ be the sequence of $N$ words comprising a document, where each $x_i \in \mathcal{W}$, the set of document words. With $\delta_z(i,j)$ as the binary indicator of whether words $x_i$ and $x_j$ are in the same segment under segmentation $z$, and $b_z(i,j)$ as the number of segment boundaries between $x_i$ and $x_j$ under $z$, the metrics are calculated as
\vspace{-2mm}
\begin{equation} \label{eq_pk}
    P_k = \frac{1}{N-\ell} \sum_{i=1}^{N-\ell} \mathbbm{1} \{ |\delta_{{\tt hyp}}(i, i+\ell) - \delta_{{\tt ref}}(i, i+\ell)| > 0 \},
\end{equation}
\begin{equation} \label{eq_wd}
    \mbox{WD} = \frac{1}{N-\ell} \sum_{i=1}^{N-\ell} \mathbbm{1} \{ |b_{{\tt hyp}}(i, i+\ell) - b_{{\tt ref}}(i, i+\ell)| > 0 \},
\end{equation}
where the window size $\ell$ is set to one less than half the average segment length, and $\mathbbm{1}$ is the indicator function which evaluates to $1$ if the argument is true and $0$ otherwise.

\vspace{-2mm}
\subsubsection{Hyperparameter analysis}
We now test $w$ and $d$ against the WD metric to validate our choices on the topic modeling task. Looking at Table \ref{fig:wd_WD}, we can see that our results are consistent with those in Section \ref{sssec:modeling_hyperparameter}: the optimal values are $d \approx 0.7$, $w = 3$ or $5$ for each of the datasets. Taken together, all of these results support the idea that we can use this set of parameters as default, across different datasets and evaluation metrics. This allows users to use BATS ``out of the box'' without fine-tuning which may be difficult due to limited data in the ``new and single document'' setting or computationally burdensome when data is available.

\vspace{-2mm}
\subsubsection{Results and discussion}
The results obtained by each algorithm on each of the four datasets are given in Table \ref{table_seg}. The mean and standard deviation across documents is shown in each case.

\bgroup
\def\arraystretch{1.2}
\begin{table*}[t]
\small
\centering
  \captionsetup{font=small}
    \caption{Text segmentation evaluation metrics obtained on \qq{six} datasets. Lower scores are better for both metrics. Our method outperforms the baselines in most cases (particularly for the Lectures dataset), except for BERTSeg and C99 on the Choi dataset \qq{and C99 on the Wiki dataset.}}
\vspace{-3mm}
\begin{threeparttable}[b]
\resizebox{1.0\linewidth}{!}{
\begin{tabular}{c||c|c||c|c||c|c||c|c}
\hline
\hline
\multirow{2}{*}{} & \multicolumn{2}{c||}{Textbook Dataset} & \multicolumn{2}{c||}{Lectures Dataset}  & \multicolumn{2}{c||}{Introductions Dataset} & \multicolumn{2}{c}{Choi Dataset}       \\ \cline{2-9}
& $P_k$ & WD & $P_k$ & WD & $P_k$ & WD & $P_k$ & WD \\
\hline
TextTiling & $0.45 \pm 0.177$ & $0.47 \pm 0.169$ & $0.43 \pm 0.099$ & $0.48 \pm 0.088$ & $0.29 \pm 0.158$ & $0.33 \pm 0.170$ & $0.33 \pm 0.076$ & $0.34 \pm 0.075$ \\
\hline
C99 & $0.55 \pm 0.142$ & $0.79 \pm 0.206$ & $0.51 \pm 0.075$ & $0.85 \pm 0.164$ & $0.20 \pm 0.120$ & $0.29 \pm 0.184$ & $0.14 \pm 0.077$ & $0.15 \pm 0.081$ \\
\hline
LDA\_MDP & $0.52 \pm 0.161$ & $0.60 \pm 0.121$ & $0.53 \pm 0.117$ & $0.63 \pm 0.118$ & $0.51 \pm 0.142$ & $0.59 \pm 0.136$ & $0.49 \pm 0.079$ & $0.50 \pm 0.088$ \\
\hline
TopicTiling & $0.50 \pm 0.157$ & $0.56 \pm 0.140$ & $0.51 \pm 0.110$ & $0.56 \pm 0.100$ & $0.50 \pm 0.138$ & $0.57 \pm 0.122$ & $0.45 \pm 0.079$ & $0.47 \pm 0.082$ \\
\hline
SupervisedSeg & $0.45 \pm 0.157$ & $0.58 \pm 0.180$ & $0.42 \pm 0.096$ & $\mathbf{0.43 \pm 0.083}$ & $0.45 \pm 0.127$ & $0.57 \pm 0.133$ & $0.23 \pm 0.064$ & $0.24 \pm 0.065$ \\
\hline
BERTSeg & {\boldmath$0.41 \pm 0.193$} & $0.44 \pm 0.181$ & $0.43 \pm 0.148$ & $0.47 \pm 0.158$  & $0.17 \pm 0.149$ & $0.19 \pm 0.163$ & $\mathbf{0.10 \pm 0.074}$ & $\mathbf{0.11 \pm 0.076}$ \\
\hline
BATS & {$\mathbf{0.41  \pm 0.162}$} & {\boldmath$0.43 \pm 0.164$} & {\boldmath$0.38 \pm 0.112$} & {\boldmath$0.43 \pm 0.112$} & {\boldmath$0.16 \pm 0.136$} & {\boldmath$0.18 \pm 0.143$} & $0.22 \pm 0.103$ & $0.23 \pm 0.109$ \\
\hline
\hline
\multirow{2}{*}{} & \multicolumn{2}{c||}{\qq{Wiki Dataset}} & \multicolumn{2}{c||}{\qq{News Dataset}}     \\ \cline{1-5} 
\cline{1-5}
& \qq{$P_k$} & \qq{WD} & \qq{$P_k$} & \qq{WD}\\
\cline{1-5}
\qq{TextTiling}  & $0.39 \pm 0.123$  & $0.41 \pm 0.122$  & $0.42 \pm 0.154$  & $0.44 \pm 0.153$ & \multicolumn{4}{c}{}\\
\cline{1-5}
\qq{C99} & $\mathbf{0.37 \pm 0.126}$ & $\mathbf{0.38 \pm 0.124}$ & $0.39 \pm 0.157$ & $0.40 \pm 0.158$  & \multicolumn{4}{c}{}\\
\cline{1-5}
\qq{LDA\_MDP} & $0.62 \pm 0.124$ & $0.68 \pm 0.136$ & $0.64 \pm 0.142$ & $0.69 \pm 0.130 $& \multicolumn{4}{c}{}\\
\cline{1-5}
\qq{TopicTiling} & $0.46 \pm 0.141 $& $0.48 \pm 0.146$ & $0.43 \pm 0.154$ & $0.45 \pm 0.153$& \multicolumn{4}{c}{}\\
\cline{1-5}
\qq{SupervisedSeg} & $0.42 \pm 0.145$ & $0.44 \pm 0.147$ & $0.41 \pm 0.151$ & $0.43 \pm 0.154$& \multicolumn{4}{c}{}\\
\cline{1-5}
\qq{BERTSeg} & $0.40 \pm 0.149$ & $0.41 \pm 0.152$ & $0.36 \pm 0.113$ & $0.38 \pm 0.114$& \multicolumn{4}{c}{}\\
\cline{1-5}
\qq{BATS}  & $0.39 \pm 0.148$ & $0.41 \pm 0.145$ & $\mathbf{0.34 \pm 0.155}$ & $\mathbf{0.36 \pm 0.157}$& \multicolumn{4}{c}{}\\
\cline{1-5}
\end{tabular}
}
\vspace{-2mm}
\end{threeparttable}
\label{table_seg}
\end{table*}
\egroup

Overall, we see that \textit{our method BATS consistently outperforms all of the baselines in terms of text segmentation on \qq{four of out of six} datasets}. BERTSeg is the most competitive on the baseline on all of the datasets except for Lectures, where SupervisedSeg does slightly better\qq{, and Wiki where C99 performs better}. The three topic-based text segmentation methods (LDA\_MDP, and TopicTiling) actually perform considerably worse than these other baselines, possibly due to single documents containing insufficient data for training their topic models (recall in particular that LDA had poor topic diversity performance in Table \ref{table_seg}). Although for the Textbook and Introductions datasets BATS achieves only modest gains over BERTSeg, this achievement is noteworthy as BERT is very large (around 335 million parameters) and, even pre-trained, has a runtime two orders of magnitude worse than BATS (see Section \ref{ssec:scala}). The performance gains are more pronounced in the Lectures \qq{and News datasets} where BATS improves more substantially on BERTSeg.

The pre-trained models outperforming the simpler baselines is perhaps not surprising. However, the fact that they achieve comparable but slightly worse results than BATS across three natural datasets suggests that (i) given their computational complexity, pre-trained embedding models are not as well suited for the ``new and single document'' setting as BATS, and (ii) the relatively simple and efficient pre-processing and biclustering techniques of BATS in combination are effective. Considering these results along with those in in Sec. \ref{ssec:res-topic}, we conclude that \textit{our method is capable of identifying accurate segment boundaries and topic words for a single document simultaneously}.


The C99 baseline performs remarkably well on the Choi \qq{and Wiki datasets}. C99 is designed specifically with datasets such as Choi in mind, where documents are artificially built with identical numbers of short segments and sparse content in each segment \cite{choi2000advances}. Specifically, as shown in Table \ref{tabD}, the average sentences per document in Choi are significantly smaller than the other datasets. This is due to the way it is constructed -- with each document as combinations of first $\ell$ sentences from documents in another corpus -- making it less realistic than the other datasets. Similarly, because BERTSeg uses agglomerative clustering on pre-trained embeddings directly (instead of using embeddings as input to another supervised model), it may be particularly apt for identifying breaks among unrelated topics (i.e., topics from different articles) as opposed to the harder task of segmenting natural text with related topics. Put simply, using the embeddings directly makes a ``coarse'' separation between unrelated topics easier as word embeddings are designed to capture these semantic differences. By contrast, BATS is designed to segment natural text with a ``finer'' separation between topics that are related but distinct in the ``single and new document'' setting. In this setting, even contextualized embeddings are not as effective as the semantic differences between topics is likely to be relatively small and so BATS outperforms embeddings-based models, as evidenced by results on the other datasets. Similar results have been observed in, e.g., \cite{koshorek2018text}.

\vspace{-2mm}
\subsection{Scalability Analysis}
\label{ssec:scala}
\vspace{-.6mm}
Next, we evaluate the effect of the number of sentences and segments on the runtime of our method compared with the baselines. Table \ref{tab:absolute_running_time} shows the increase in runtime from varying the number of sentences in each segment for the Choi dataset, relative to the case of 50 sentences (we choose this dataset because all documents are constructed with a constant number of segments). We can see that \textit{the growth in runtime of our methodology BATS is comparable to the most scalable baselines}, with the rate of increase in runtime less than the corresponding increase in sentences. Additionally, BATS is the only methodology performing both topic modeling and text segmentation. Out of the baselines in Table \ref{tab:absolute_running_time}, TextTiling, C99, and pLSA have considerably higher increases in runtime, with pLSA performing the worst. The substantial difference between LSA, the most scalable, and pLSA is consistent with spectral approaches being known to scale better than generative algorithms that require multiple iterations \cite{zhong2005generative}. Although BERTSeg and SupervisedSeg scale relatively well, they have absolute runtimes second and third to pLSA for all sentence lengths. This indicates that there is a high fixed overhead associated with loading the large pre-trained models for segmentation.





\begin{table}[t]
\vspace{-2mm}
\centering
\caption{Absolute running time and increasing rate when varying the number of sentences in the Choi dataset. The time increase is relative to the case of 50 sentences. Our method scales well compared with the baselines.}
\vspace{-3mm}
\label{tab:absolute_running_time}
\begin{adjustbox}{width=1.0\textwidth}
\begin{tabular}{l|l|ll|ll|ll|ll|ll}
\hline
Sentences & \multicolumn{1}{c|}{50 (baseline)} & \multicolumn{2}{c|}{60} & \multicolumn{2}{c|}{80} & \multicolumn{2}{c|}{100} & \multicolumn{2}{c|}{120} & \multicolumn{2}{c}{140} \\ \hline
Runtime              & Abs. & Abs.     & Rate & Abs.     & Rate & Abs.     & Rate & Abs.     & Rate & Abs.      & Rate \\ \hline
LSA           & 5.21     & 6.15     & 1.18 & 7.24     & 1.39 & 8.53     & 1.64 & 9.54     & 1.83 & 10.62     & 2.04 \\
NMF           & 35.63    & 43.72    & 1.23 & 56.81    & 1.59 & 64.65    & 1.81 & 86.71    & 2.43 & 92.79     & 2.60 \\
HDP           & 89.6     & 109.12   & 1.22 & 137.78   & 1.54 & 179.55   & 2.00 & 202.35   & 2.26 & 238.98    & 2.67 \\
TopicTiling   & 55.02    & 67.63    & 1.23 & 87.69    & 1.59 & 109.95   & 2.00 & 125.56   & 2.28 & 148.61    & 2.70 \\
LDA\_MDP      & 931.26   & 1130.78  & 1.21 & 1446.27  & 1.55 & 1843.17  & 1.98 & 2022.65  & 2.17 & 2598.78   & 2.79 \\
LDA           & 837.8    & 1078.98  & 1.29 & 1342.06  & 1.60 & 1710.25  & 2.04 & 2085.36  & 2.49 & 2395.94   & 2.86 \\
TextTiling    & 90.12    & 140.18   & 1.56 & 227.75   & 2.53 & 349.76   & 3.88 & 447.13   & 4.96 & 639.63    & 7.10 \\
C99           & 28.42    & 44.32    & 1.56 & 74.21    & 2.61 & 119.24   & 4.20 & 172.44   & 6.07 & 227.15    & 7.99 \\
pLSA          & 14742.62 & 23388.19 & 1.59 & 39804.17 & 2.70 & 64856.58 & 4.40 & 90536.42 & 6.14 & 125231.36 & 8.49 \\
SupervisedSeg  & 2103.45  & 2615.49  & 1.24 & 2676.45  & 1.27 & 3419.84  & 1.63 & 4350.65  & 2.07 & 4727.54   & 2.25 \\
BERTSeg    & 5361.82  & 7616.14  & 1.42 & 8522.98  & 1.59 & 9663.25  & 1.80 & 12324.16 & 2.30 & 13193.24  & 2.46 \\
SeaNMF        & 79.46    & 85.67    & 1.08 & 114.27   & 1.44 & 140.79   & 1.77 & 152.14   & 1.91 & 206.71    & 2.60 \\
BATS          & 52.49    & 60.29    & 1.15 & 69.71    & 1.33 & 89.84    & 1.71 & 95.75    & 1.82 & 109.62    & 2.09 \\\hline
\end{tabular}
\end{adjustbox}
\end{table}

\begin{table}[ht!]
\vspace{-2mm}
\centering
\caption{\small{The absolute running time and increasing rate for text segments in the Lectures dataset with the varying number of segments. The baseline is 10 segments, and each bar is over 10 runs.}}
\vspace{-3mm}
\begin{adjustbox}{width=0.8\textwidth}
\label{tab:runtime_lectures}
\begin{threeparttable}
\begin{tabular}{l|ll|ll|ll|ll|ll}
\hline
Segments &
  \multicolumn{2}{c|}{10 (baseline)} &
  \multicolumn{2}{c|}{20} &
  \multicolumn{2}{c|}{30} &
  \multicolumn{2}{c}{40} &
  \multicolumn{2}{c}{50} \\ \hline
Runtime &
  \multicolumn{1}{c}{Abs.} &
  \multicolumn{1}{c|}{Rate} &
  \multicolumn{1}{c}{Abs.} &
  \multicolumn{1}{c|}{Rate} &
  \multicolumn{1}{c}{Abs.} &
  \multicolumn{1}{c|}{Rate} &
  \multicolumn{1}{c}{Abs.} &
  \multicolumn{1}{c|}{Rate} &
  \multicolumn{1}{c}{Abs.} &
  \multicolumn{1}{c}{Rate} \\ \hline
HDP &
  1105.63 &
  1 &
  1104.4 &
  1 &
  1104.12 &
  1 &
  1104.41 &
  1 &
  1105.45 &
  1 \\
LDA &
  4272.27 &
  1 &
  4272.86 &
  1 &
  4319.22 &
  1.01 &
  4381.64 &
  1.03 &
  4453.33 &
  1.04 \\
LSA &
  225.94 &
  1 &
  266.43 &
  1.18 &
  289.98 &
  1.28 &
  303.22 &
  1.34 &
  312.37 &
  1.38 \\
NMF &
  491.57 &
  1 &
  570 &
  1.16 &
  647.43 &
  1.32 &
  739.32 &
  1.5 &
  840.84 &
  1.71 \\
SeaNMF  &
  526.38 &
   1 &
  950.87 &
  1.81 &
  1211.23 &
  2.3 &
  1370.74 &
  2.6 &
  2299.02 &
  4.37 \\
  BATS &
  2434.27 &
  1 &
  2448.87 &
  1.01 &
  2469.49 &
  1.01 &
  2493.14 &
  1.02 &
  2517.15 &
  1.03 \\\hline \hline
\multicolumn{1}{c|}{Segments} &
  \multicolumn{2}{c|}{60} &
  \multicolumn{2}{c|}{70} &
  \multicolumn{2}{c|}{80} &
  \multicolumn{2}{c}{90} &
  \multicolumn{2}{c}{100} \\ \hline
Runtime &
  \multicolumn{1}{c}{Abs.} &
  \multicolumn{1}{c|}{Rate} &
  \multicolumn{1}{c}{Abs.} &
  \multicolumn{1}{c|}{Rate} &
  \multicolumn{1}{c}{Abs.} &
  \multicolumn{1}{c|}{Rate} &
  \multicolumn{1}{c}{Abs.} &
  \multicolumn{1}{c|}{Rate} &
  \multicolumn{1}{c}{Abs.} &
  \multicolumn{1}{c}{Rate} \\ \hline
HDP &
  1108.22 &
  1 &
  1111.68 &
  1.01 &
  1115.45 &
  1.01 &
  1119.19 &
  1.01 &
  1122.55 &
  1.02 \\
LDA &
  4518.03 &
  1.06 &
  4576.03 &
  1.07 &
  4629.34 &
  1.08 &
  4678.3 &
  1.1 &
  4722.8 &
  1.11 \\
LSA &
  316.39 &
  1.4 &
  320.22 &
  1.42 &
  323.86 &
  1.43 &
  327.44 &
  1.45 &
  330.99 &
  1.46 \\
NMF &
  888.14 &
  1.81 &
  934.04 &
  1.9 &
  982.02 &
  2 &
  1033.77 &
  2.1 &
  1090.05 &
  2.22 \\
SeaNMF  &
  3167.22 &
  6.02 &
  4235.9 &
  8.05 &
  6324.47 &
  12.02 &
  7827.73 &
  14.87 &
  8692.67 &
  16.51 \\
BATS &
  2540.81 &
  1.04 &
  2564.17 &
  1.05 &
  2587.01 &
  1.06 &
  2609.59 &
  1.07 &
  2632.07 &
  1.08 \\\hline
\end{tabular}
\end{threeparttable}
\end{adjustbox}
\end{table}

Table \ref{tab:runtime_lectures} shows the impact on runtime from varying the number of segments per document for the Lectures dataset (recall from Table \ref{tabD} this dataset has the longest documents available). Here we have excluded pLSA, as its runtime is significantly longer, and also the text segmentation baselines, as their runtimes are not dependent on the number of segments. SeaNMF is by far impacted the most, followed by NMF and LSA, while HDP and LDA exhibit the best scalability. Our method remains impacted under 10\% for a $5$-fold increase in segments, again implying that \textit{our method supports changes in the size of input efficiently}. We note that while for a small number of topics BATS has a longer absolute runtime than SeaNMF, this disadvantage disappears when the number of segments exceeds 50. Taken together, Tables \ref{tab:absolute_running_time} and \ref{tab:runtime_lectures} validate our theoretical analysis in Section \ref{ssec:complexity} which concluded that BATS has low computational complexity.

When we consider these runtime results in the context of the evaluation metrics, we find that \textit{BATS provides the best balance of computational efficiency and performance across evaluation metrics for both the text segmentation and the topic modeling tasks.} When compared with the most competitive topic modeling baselines (LSA and SeaNMF), BATS is not as efficient as LSA but scores much better on topic similarity measures. BATS scales better than SeaNMF and has a significantly lower runtime on longer documents with many segments, as evidenced by its performance on the Lectures dataset (Table \ref{tab:runtime_lectures}). On the text segmentation task, BATS both scales better and has a much lower absolute running time than the two best baselines (BERTSeg and SupervisedSeg) in addition to outperforming both of them on three real-world datasets. These results show that BATS is a practical algorithm with substantial advantages over existing baselines for the ``single and new document'' setting that we are focused on in this paper.

\vspace{-2mm}
\subsection{Ablation Study}
\label{ssec:ablation}
To test the effectiveness of different parts of the BATS methodology, we conduct an ablation study that excludes components and measures the resulting performance effect. The results for both the topic coherence and text segmentation metrics are given in Table \ref{tab:ablation_study}. Referring to Figure \ref{fig:framework}, we focus on the following parts of BATS: (i) parts of speech (POS) awarding (i.e., setting $\lambda = 0$ in (\ref{eq:pos_awarding})), (ii) low frequency processing (i.e., not excluding degree one words), (iii) regularization of the graph Laplacian (i.e., setting the $\tau$ terms to $0$ in (\ref{eq:laplacian})), (iv) sentence bonding (i.e., setting $d = 0$ in (\ref{eq:sentence_bonding}))\qq{, and (v) the Laplacian. For (v), we remove the Laplacian entirely, and instead perform SVD directly on the sentence-word matrix after POS awarding and sentence bonding}.

We draw three main conclusions from these results. First, each component of BATS has a substantial impact on one or more of the metrics, which validates its inclusion in the methodology. Second, the topic coherence metrics are impacted by each component more significantly than the text segmentation metrics, which implies that results at the sentence level are not impacted as significantly by preprocessing than those at the topic level. Third, out of all the components, the \qq{Laplacian has the greatest impact, followed by sentence bonding. This validates our graph Laplacian representation as input to the spectral clustering algorithm and suggests that our connection to denoising under the stochastic block model is appropriate. The impact of removing sentence bonding reflects} the importance of incorporating word order information in the presence of sparsity.
\begin{table}[t]
\vspace{-2mm}
\caption{\small{Ablation studies on the Textbook, Lectures, and Introduction datasets. Each row gives the performance on metrics obtained when excluding a component of the methodology. Overall, we see that each component is important to BATS, as excluding any of them results in lower performance on one or more metrics.}}
\vspace{-4mm}
\small
\centering
\begin{adjustbox}{width=1.0\textwidth}
\begin{tabular}{c||ccll||ccll}
\hline
&\multicolumn{4}{c||}{Textbook Dataset} & \multicolumn{4}{c}{Lectures Dataset}                                                                                                                                                       \\ \hline
\multicolumn{1}{c||}{}                               & \multicolumn{1}{c|}{PMI}             & \multicolumn{1}{c|}{UMass}             & \multicolumn{1}{l|}{$P_k$}               & WD & \multicolumn{1}{c|}{PMI}             & \multicolumn{1}{c|}{UMass}             & \multicolumn{1}{l|}{$P_k$}               & WD                \\ \hline
\multicolumn{1}{l||}{BATS - SPEECH}                  & \multicolumn{1}{c|}{$0.42\pm 0.74$} & \multicolumn{1}{c|}{$ -11.21 \pm 2.44 $}  & \multicolumn{1}{l|}{$0.42 \pm 0.18$}  & $0.44 \pm 0.17$  & \multicolumn{1}{c|}{$0.29 \pm 0.55$} & \multicolumn{1}{c|}{$-9.88 \pm 2.05$}  & \multicolumn{1}{l|}{$0.42 \pm 0.16$}  & $0.44 \pm 0.16$  \\
\multicolumn{1}{l||}{BATS - FREQUENCY}        & \multicolumn{1}{c|}{$0.41 \pm 0.75$} & \multicolumn{1}{c|}{$-11.18 \pm 2.88$} & \multicolumn{1}{l|}{$0.42  \pm 0.17$} & $0.44 \pm 0.16$  & \multicolumn{1}{c|}{$0.21 \pm 0.58$} & \multicolumn{1}{c|}{$-10.79 \pm 2.11$} & \multicolumn{1}{l|}{$0.41 \pm 0.11$}  & $0.45 \pm 0.11$  \\
\multicolumn{1}{l||}{BATS - REGULARIZATION} & \multicolumn{1}{c|}{$-2.25 \pm 0.77$} & \multicolumn{1}{c|}{$-14.81 \pm 1.83$} & \multicolumn{1}{l|}{$0.44  \pm 0.18$} & $0.45 \pm 0.17$  & \multicolumn{1}{c|}{$-2.27\pm 0.97$}  & \multicolumn{1}{c|}{$-14.99 \pm 1.28$} & \multicolumn{1}{l|}{$0.40  \pm 0.12$} & $0.44 \pm 0.12$  \\
\multicolumn{1}{l||}{BATS - BONDING} & \multicolumn{1}{c|}{$-2.38 \pm 0.76$} & \multicolumn{1}{c|}{$-14.69 \pm 1.79$} & \multicolumn{1}{l|}{$0.44  \pm 0.19$} & $0.47 \pm 0.17$  & \multicolumn{1}{c|}{$-2.39 \pm 0.99$} & \multicolumn{1}{c|}{$-14.80 \pm 1.16$} & \multicolumn{1}{l|}{$0.42  \pm 0.13$} & $0.47 \pm 0.12$  \\
\multicolumn{1}{l||}{\qq{BATS - LAPLACIAN}} & \multicolumn{1}{c|}{$-3.46\pm 0.82$} & \multicolumn{1}{c|}{$ -16.48 \pm 2.89 $}  & \multicolumn{1}{l|}{$0.48 \pm 0.18$}  & $0.51 \pm 0.18$ & \multicolumn{1}{c|}{$-3.79\pm 0.54$} & \multicolumn{1}{c|}{$ -16.11 \pm 1.82 $}  & \multicolumn{1}{l|}{$0.51 \pm 0.14$}  & $0.54 \pm 0.15$    \\
\multicolumn{1}{l||}{BATS}                           & \multicolumn{1}{c|}{$\mathbf{0.62 \pm 0.63}$} & \multicolumn{1}{c|}{$\mathbf{-10.28 \pm 2.31}$} & \multicolumn{1}{l|}{$\mathbf{0.41  \pm 0.16}$} & $\mathbf{0.43 \pm 0.16}$  & \multicolumn{1}{c|}{$\mathbf{0.54 \pm 0.53}$} & \multicolumn{1}{c|}{$\mathbf{-9.08 \pm 1.82}$}  & \multicolumn{1}{l|}{$\mathbf{0.38 \pm 0.11}$}  & $\mathbf{0.43 \pm 0.11}$  \\ \hline \hline
&\multicolumn{4}{c||}{Introductions Dataset} & \multicolumn{4}{c}{}                                                                                                                                                        \\ \cline{1-5}
\multicolumn{1}{c||}{}                               & \multicolumn{1}{c|}{PMI}             & \multicolumn{1}{c|}{UMass}             & \multicolumn{1}{l|}{$P_k$}               & WD & \multicolumn{1}{c}{}             & \multicolumn{1}{c}{}             & \multicolumn{1}{l}{}               &                 \\ \cline{1-5}
\multicolumn{1}{l||}{BATS - SPEECH}  & \multicolumn{1}{c|}{$0.17 \pm 0.45$} & \multicolumn{1}{c|}{$-8.18 \pm 1.62$}  & \multicolumn{1}{l|}{$0.163 \pm 0.14$} & $0.185 \pm 0.15$  &  &   &   &  \\
\multicolumn{1}{l||}{BATS - FREQUENCY}         & \multicolumn{1}{c|}{$0.16 \pm 0.47$} & \multicolumn{1}{c|}{$-8.40 \pm 1.71$}  & \multicolumn{1}{l|}{$0.164 \pm 0.15$} & $0.186 \pm 0.16$  & & &   &\\
\multicolumn{1}{l||}{BATS - REGULARIZATION} & \multicolumn{1}{c|}{$0.15\pm 0.49$}  & \multicolumn{1}{c|}{$-12.09 \pm 2.21$}  & \multicolumn{1}{l|}{$0.164  \pm 0.14$} & $0.189 \pm 0.15$  & &  & & \\
\multicolumn{1}{l||}{BATS - BONDING} & \multicolumn{1}{c|}{$-2.25 \pm 0.77$} & \multicolumn{1}{c|}{$-14.81 \pm 1.83$} & \multicolumn{1}{l|}{$0.166  \pm 0.15$} & $0.188 \pm 0.17$  & &  &  &\\
\multicolumn{1}{l||}{\qq{BATS - LAPLACIAN}}  & \multicolumn{1}{c|}{$-2.67\pm 0.73$} & \multicolumn{1}{c|}{$ -15.21 \pm 2.12 $}  & \multicolumn{1}{l|}{$0.376 \pm 0.10$}  & $0.382 \pm 0.11$ & & & & \\
\multicolumn{1}{l||}{BATS}                           & \multicolumn{1}{c|}{$\mathbf{0.58 \pm 0.45}$} & \multicolumn{1}{c|}{$\mathbf{-7.69 \pm 1.63}$}  & \multicolumn{1}{l|}{$\mathbf{0.160 \pm 0.14}$} & $\mathbf{0.180 \pm 0.14}$  &  &   &   & \\
\cline{1-5}
\end{tabular}
\end{adjustbox}
\label{tab:ablation_study}
\end{table}
\vspace{-1mm}
\section{Conclusion and Future Work}
The ``single and new document'' setting arises when it is desirable to perform modeling tasks on a single document, due to the potential novelty of words in the document and/or computational limitations. In this work, we developed an unsupervised, computationally efficient, statistically sound methodology called BATS that simultaneously extracts the topics and segments the text from one single document. BATS first leverages word-order information together with optimization tricks such as parts-of-speech (POS) tagging to refine a document's sentence-word matrix. It then obtains a singular value decomposition from a regularized form of the graph Laplacian, with the singular vectors yielding low dimensional embeddings of words and sentences. Finally, BATS employs clustering algorithms to extract topics and text segments from the left and right singular vectors. Through evaluations against 12 baselines on \qq{six} datasets, we confirmed that our algorithm achieves the best overall results considering runtime, scalability, and standard metrics in both topic modeling and text segmentation tasks for the ``single and new document'' setting. For topic modeling, this was especially true when considering the dual objectives of coherence maximization and similarity minimization across topics. Our experimental results also showed that BATS scales well with the size of the input data, and that it is robust to changes in dataset characteristics.

We identify several potential avenues of future work. First, \qq{BATS could potentially be tuned to improve performance on the topic modeling or text segmentation task at the cost of increased computation. Specifically, although we found that BATS outperformed baselines based on word embedding techniques, further experimentation may leverage embeddings to enhance BATS. As this would likely come at substantial computational cost (see Tables \ref{tab:absolute_running_time} and \ref{tab:runtime_lectures}), we leave this as an avenue for future work beyond the ``single and new document'' setting.} Second, a more elaborate POS awarding scheme in the sentence-word matrix construction phase of BATS may improve topic coherence further. Third, since BATS provides both topic and text segment information, the application of our methodology to text summarization can also be considered, e.g., in identifying the most important segments according to the number of corresponding topic words. 

\newpage

\section{Acknowledgement}
We thank anonymous reviewers for helpful comments and suggestions. Christopher G. Brinton was supported in part by the Charles Koch Foundation. Qiong Wu and Zhenming Liu are supported by NSF grants NSF-2008557, NSF-1835821, and NSF-1755769. Yanhua Li was supported in part by NSF grants IIS-1942680 (CAREER), CNS-1952085, CMMI1831140, and DGE-2021871.
The authors acknowledge William $\&$ Mary Research Computing for providing computational resources and technical support that have contributed to the results reported within this paper.

\bibliographystyle{ACM-Reference-Format}
\bibliography{bib}


\end{document}